\documentclass[journal,comsoc]{IEEEtran}
\usepackage[T1]{fontenc}
\usepackage{amsmath,textcomp,array,cite}
\usepackage{mathtools}
\usepackage{graphicx}
\usepackage[export]{adjustbox}
\usepackage{hyperref}
\usepackage{subfig}
\usepackage{tabularx}
\usepackage{verbatim}
\usepackage{cuted}
\usepackage{algorithmic,amsbsy,amsmath,amssymb,epsfig,bbm,mathrsfs,multirow,amsthm}
\usepackage[english]{babel}
\usepackage[cmintegrals]{newtxmath}
\usepackage[font={footnotesize},justification=centering]{caption}
\usepackage{nccmath }
\usepackage{color}
\usepackage{pifont}
\usepackage{lipsum}
\usepackage[noend,ruled,vlined]{algorithm2e}
\hypersetup{pdfstartview={XYZ null null 1.00}}
\SetAlFnt{\footnotesize}
\SetAlCapFnt{\small}
\algsetup{linenosize=\tiny}
\usepackage[outdir=./]{epstopdf}

\begin{document}
\makeatletter
\renewcommand{\fnum@figure}{Fig. \thefigure}
\renewcommand{\@algocf@capt@plain}{above}
\newtheorem{theorem}{Theorem}
\newtheorem{remark}{Remark}
\makeatother
\title{Backscatter-Assisted Wireless Powered Communication Networks Empowered by Intelligent Reflecting Surface}
\author{Parisa Ramezani, \textit{Graduate Student Member, IEEE}, and Abbas Jamalipour, \textit{Fellow, IEEE}
    \thanks{The authors are with the School of Electrical and Information Engineering, University of Sydney, NSW 2006, Australia, (e-mails: parisa.ramezani@sydney.edu.au, a.jamalipour@ieee.org).}
      \thanks{Copyright (c) 2021 IEEE. Personal use of this material is permitted. However, permission to use this material for any other purposes must be obtained from the IEEE by sending a request to pubs-permissions@ieee.org.}  
   }

\maketitle
\begin{abstract} Intelligent reflecting surface (IRS) has been recently emerged as an effective way for improving the performance of wireless networks by reconfiguring the propagation environment through a large number of passive reflecting elements. This game-changing technology is especially important for stepping into the Internet of Everything (IoE) era, where high performance is demanded with very limited available resources. In this paper, we study a backscatter-assisted wireless powered communication network (BS-WPCN), in which a number of energy-constrained users, powered by a power station (PS), transmit information to an access point (AP) via backscatter and active wireless information transfer, with their communication being aided by an IRS. Using a practical energy harvesting (EH) model which is able to capture the characteristics of realistic energy harvesters, we investigate the maximization of total network throughput. Specifically, IRS reflection coefficients, PS transmit and AP receive beamforming vectors, power and time allocation are designed through a two-stage algorithm, assuming minimum mean square error (MMSE) receiver at the AP. The effectiveness of the proposed algorithm is confirmed via extensive numerical simulations. We also show that our proposed scheme is readily applicable to practical IRS-aided networks with discrete phase shift values.  
\end{abstract}
\begin{IEEEkeywords} Wireless powered communication network, backscatter communication, intelligent reflecting surface, minimum mean square error.
\end{IEEEkeywords}
\section{Introduction}
Enabling devices to obtain their required energy for communication in a self-sustainable manner, wireless powered communication network (WPCN) is regarded as an inevitable component of the much-anticipated Internet of Everything (IoE) \cite{mag1,mag2,mag3}. Yet, like any other technology, WPCN has its shortcomings. WPCN devices need sufficient time to harvest and accumulate their required energy before being able to take part in communication. The energy accumulation time can get longer if the conditions are not favorable for harvesting energy. Longer energy harvesting (EH) time translates to a shorter time being remained for information transmission (IT), which negatively impacts the network performance. 

As another important enabler for IoE, backscatter communication allows devices to transmit information to their intended receivers without needing to generate active radio frequency (RF) signals. In backscatter transmission, devices modulate their information onto the impinging RF signals and reflect them toward their intended receiver. Backscatter communication is not the only technology which operates based on the reflection of wireless signals. Recently, intelligent reflecting surface (IRS) has come on the scene with the unprecedented capability of modifying the propagation environment. IRS is composed of a large number of low-cost passive reflecting elements which can be adaptively configured to impose phase and amplitude changes on the received signals and reflect them towards the desired direction. Both backscatter communication and IRS have great potentials to improve the performance and alleviate the shortcomings of WPCNs. 

%The modification applied by IRS can compensate for poor channel conditions and allow devices to communicate even if there exists blockage between the transmitter and the receiver. 

Note that although both backscatter communication and IRS operate based on the reflection of the signals generated by other RF sources, there exists a fundamental difference between these two technologies. Particularly, IRS is mainly employed as an auxiliary entity which assists the communication between network elements by modifying and reflecting the incident signals, while a backscattering device modulates its own information onto the incoming signals and scatters the signals to its designated receiver \cite{irs11}. 
\subsection{Background}
The term WPCN was first introduced in \cite{wpcn0}, where the authors studied the uplink IT of a number of users to a hybrid access point (HAP) which also served the role of a downlink energy transmitter for the users. Since then, WPCN has gained considerable attention as an indispensable building block for the realization of the self-sustaible IoE. Extensive research has been done for extending WPCNs and improving their performance, among which we can name designing energy beamforming vectors \cite{wpcn_eb1,wpcn_eb2}, adding full-duplex (FD) operation for simultaneous downlink energy and uplink information transfer \cite{wpcn_fd}, using power beacons (PBs) to enhance the efficiency of wireless energy transfer (WET)\cite{wpcn_pb}, and integrating this technology with relay-based communication networks \cite{wpcn_relay2,wpcn_relay3,wpcn_relay4}. Readers can refer to \cite{mag1,mag2,mag3} for more details on the fundamentals of WPCNs. 

Integration of other technologies into WPCNs is another promising way for improving the performance of these networks and making them more efficient. In this regard, backscatter communication is an apt candidate which can complement the conventional wireless powered communication and notably enhance the performance of WPCNs. As backscatter communication consumes much lower energy than the active wireless powered communication, the instantaneous harvested energy of backscattering devices is sufficient to power their information transmission.  Therefore, using hybrid-mode radios which can switch between passive backscatter and active wireless powered transmissions has been proposed in order to allow users to choose between the two communication modes based on their channel and energy conditions \cite{rev1,rev2}. In our previous work, we have studied the integration of backscatter communication into WPCN, where the WET phase of the traditional WPCN model has been modified to let users use the energy signals for the dual purposes of EH and backscattering \cite{bs2}.  Reference \cite{bs4} provides a comprehensive survey on backscatter communication and the benefits of integrating this technology with WPCN.    

IRS, consisting of a large number of low-cost reflecting elements, has recently emerged as a revolutionary solution to improve the performance of wireless communication networks. IRS can modify the propagation environment and create favorable conditions for energy and information transfer without using energy-hungry RF chains. Thanks to its unparalleled features and functionalities, IRS has been recently applied to various networks and proved effective for improving the performance of wireless systems in a multifaceted manner \cite{irs1,irs3,irs4}. The integration of IRS with WPCNs has also been lately investigated in a few works. In \cite{irs10}, a TDMA-based WPCN is considered, where a self-sustainable IRS empowers the energy and information transfer between the HAP and the users. Time-switching and power-splitting schemes are studied for EH at the IRS and the authors optimize the IRS phase shifts, time and power allocation, and time switching/power splitting ratios for maximizing the total throughput. Reference \cite{rev3} proposes an IRS-assisted WPCN which operates based on a novel hybrid non-orthogonal multiple access (NOMA) and TDMA protocol. The users are grouped into different clusters and the TDMA strategy is applied for the transmission of different clusters, while the users in the same cluster transmit information to the HAP simultaneously using the concept of NOMA.  The readers can refer to \cite{irs12,irs13} for a profound survey on IRS-assisted systems and directions for future research in this area.

\subsection{Motivation}
 WPCN, Backscatter communication, and IRS will be key players for realizing the envisioned massive connectivity use cases in the imminent IoE era in next-generation networks \cite{Lina}. The WET-enabled communication in WPCN relieves network devices from the issues pertaining to energy limitation so that the focus can be shifted to performance optimization without being concerned about the energy shortage problem at network devices and making them use their available energy prudently. Backscatter communication also helps devices use available signals more efficiently by letting them enhance their throughput performance without consuming extra energy for active generation of RF signals. IRS is also a ground-breaking technology which can bring significant performance gains into wireless networks by enabling the dynamic modification of RF signals impinging on the surface, thus virtually refining the propagation environment and strengthening the signals at the receiver. These three innovative paradigms have been proved useful for enhancing the performance of wireless systems and paving the way towards the future self-sustainable networks. When joining together, these technologies can very well cater to the needs of massive number of power-constrained devices in the IoE era, enabling their efficient communication without burdening the network with excessive costs. This is the main motivation of this research work, where we study a backscatter-assisted WPCN (BS-WPCN) empowered by the incorporation of an IRS which assists in backscatter and active information transmission of WPCN users. 

Another key point in studying the performance of EH-enabled networks is that the model being used for EH at network devices must be able to capture the important characteristics of practical energy harvesters such as sensitivity and saturation. This is essential in order to avoid any notable mismatch between theoretical studies and real-life implementations. The conventional linear EH model used in most of the works on WPCN (e.g. \cite{wpcn0}) is not valid for disregarding the fundamental features of practical EH circuits. The well-known sigmoidal EH model \cite{elena,bruno} which accounts for the practical EH characteristics has also tractablity issues. Therefore, there is a compelling need for a simple and tractable way for modeling the behavior of EH circuits. This has motivated us to present a piece-wise linear EH model with three pieces which, though being simple, can model sensitivity and saturation effects of practical EH circuits.

\subsection{Contributions}
In this paper, we consider an IRS-empowered BS-WPCN, where the EH users are powered by WET from a power station (PS) and their information transfer to the access point (AP) is aided by an IRS. Our main contributions are summarized as follows:
\begin{itemize}
    \item We propose an IRS-empowered BS-WPCN, where multiple users transmit information to an AP via backscatter and active information transfer. The IRS assists the information transfer of the users by adjusting the reflection coefficients of its elements; a PS is also deployed to power the transmissions of energy-constrained users. To the authors' best knowledge, this is the first work to study the integration of IRS with the hybrid of active wireless powered communication and passive backscatter communication.
    \item  One of the most important design considerations in EH-enabled networks that is often overlooked is the non-linear input-output relationship of practical EH circuits. We propose to use a piece-wise linear EH model with three pieces at the users, which accounts for the sensitivity and saturation effects of practical energy harvesters. We verify the accuracy of this model by applying it to several real measurements and comparing the real measured values with the ones estimated via this model.
    \item We formulate a sum-throughput maximization problem and study the optimization of IRS reflection coefficients (including both amplitude reflection and phase shift at each element), time and power allocation, transmit beamforming at the PS, and receive beamforming at the AP, under the practical average and peak power constraints at the PS as well as the energy causality constraint at the users. Optimizing the amplitude reflection, which is mostly ignored in studies on IRS-aided systems is essential when users simultaneously transmit their information signals and the AP receives a combination of the information of all users. In such as case, the optimal amplitude reflection may be smaller than one. We propose a two-stage algorithm for solving the optimization problem, where the IRS reflection coefficients for empowering the users' backscatter communication are optimized in the first stage and the optimization of other variables is carried out in the second stage. The techniques of alternating optimization (AO), successive approximation (SA), semidefinite relaxation (SDR), and block coordinate descent (BCD) are used for optimizing the variables. 
    %\item We propose two different methods for optimizing the IRS reflection coefficients in the first stage based on classic alternating optimization (AO) and successive approximation (SA) techniques. For the second-stage optimization which involves optimizing the PS transmit beamforming, AP receive beamforming, time and power allocation, and the coefficients of the IRS when assisting users' active IT, we propose solutions based on the SA technique and the well-known semidefinite relaxation (SDR). Further, assuming minimum mean square error (MMSE) receiver at the AP, an efficient method is proposed based on the block coordinate descent (BCD) technique. As closed-form solutions are acquired in each iteration of the BCD procedure, the algorithm is guaranteed to converge to an optimal solution for AP receive beamforming and IRS reflection coefficients. 
   % \item  For the first stage which corresponds to optimizing the IRS reflection coefficients for empowering the users' backscatter communication, we propose two different algorithms based on alternating optimization (AO) and successive approximation (SA) techniques.  
    %\item For the second-stage optimization problem, we use semidefinite relaxation (SDR) and SA techniques to optimize the PS transmit beamforming vectors, power and time allocation. Also, assuming minimum mean square error (MMSE) receiver at the AP, the IRS reflection coefficients for assisting the users' active transmission and receive beamforming vectors at the AP are jointly optimized via an efficient method based on the block coordinate descent (BCD) technique. 
    \item We validate the efficiency of the proposed scheme by conducting extensive numerical simulations. We compare the performance of our proposed algorithms with different benchmark schemes and discuss the superiority of our scheme over the benchmarks.  We also show that the presented algorithm can be readily applied to the scenario with practical IRS settings, where only a limited number of phase shifts can be chosen for the IRS elements. We particularly show that a 2-bit resolution for the phase shift of IRS elements is sufficient for achieving the envisioned near-optimal performance.

\end{itemize}

\begin{figure}[t!]
\centering
\includegraphics[width=2.8in]{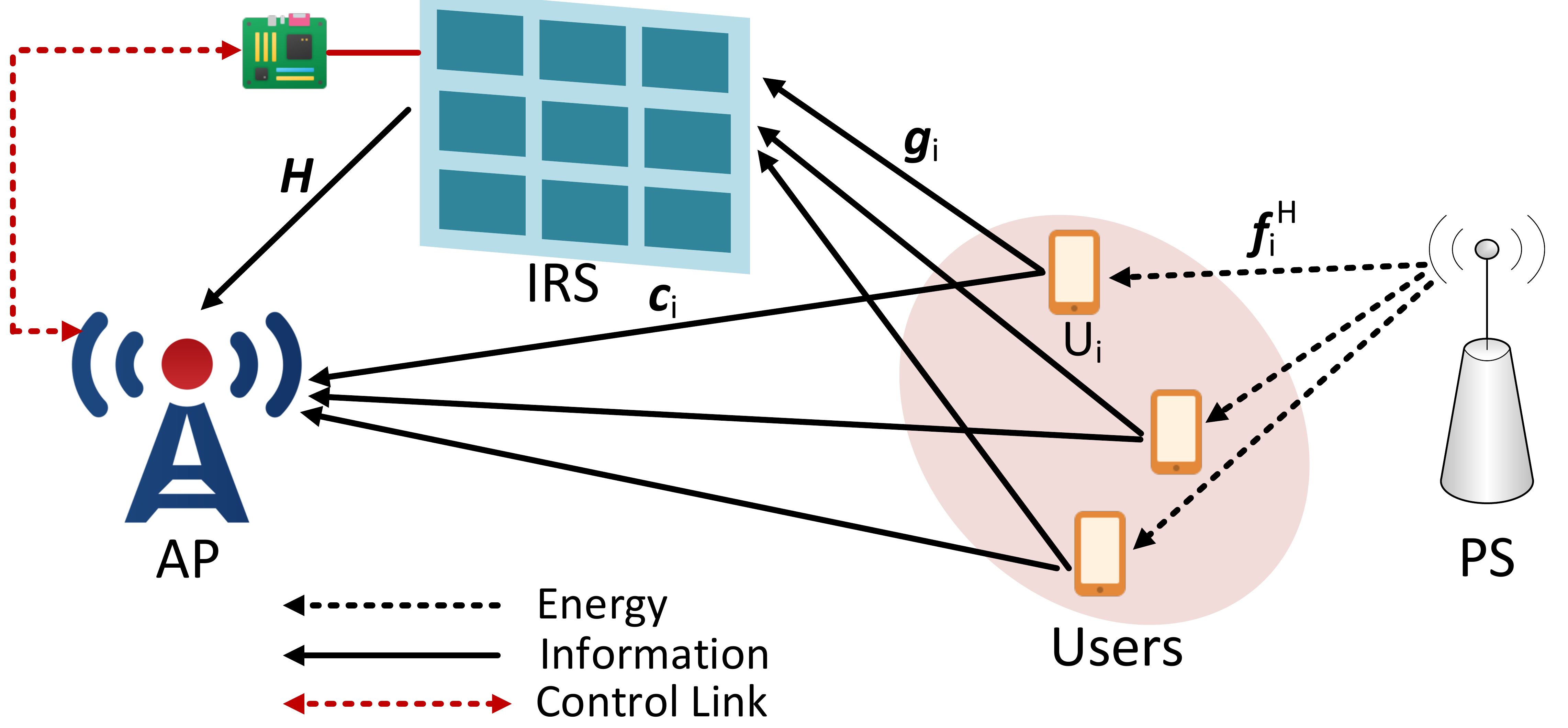}
\caption{ An IRS-empowered BS-WPCN}
	\label{fig1}
	
\end{figure}

\subsection{Organization}
\textit{Organization:}
We present the system model and formulate the throughput maximization problem in Section II. The two-stage algorithm for solving the formulated problem is elaborated in Sections III and IV. Numerical simulations for evaluating the performance of the proposed algorithms are presented in Section V, and Section VI concludes the paper.
\begin{comment}
\textit{Notations:}  Scalars are denoted by italic letters, vectors and matrices are denoted by bold-face lower-case and upper-case letters, respectively. $(\cdot)^T$, $(\cdot)^H$, and $(\cdot)^{-1}$ denote transpose, Hermitian transpose, and matrix inversion operations, while $\text{Tr}(\cdot)$ and $\text{Rank}(\cdot)$ denote the matrix trace and rank, respectively. $\mathbb{C}^{a \times b}$ is the space of $a\times b$ complex-valued matrices. $\mathbb{E}[\cdot]$ stands for expectation and $\boldsymbol{I}_M$ denotes the identity matrix of size $M$.  $\text{Re}\{\cdot\}$, $|\cdot|$, $\text{arg}(\cdot)$, and $\overline{(\cdot)}$ denote the real part, absolute value, angle, and conjugate of a complex number, respectively. $||\cdot||$ is the 2-norm of a vector and $[\cdot]_{i,j}$ indicates the element in the $i$-th row and the $j$-th column of a matrix. The distribution of a circularly symmetric complex Gaussian (CSCG) random vector is denoted by $\mathcal{C}\mathcal{N}(\boldsymbol{\mu},\boldsymbol{\Sigma})$ with $\boldsymbol{\mu}$ being the mean vector and $\boldsymbol{\Sigma}$ representing the covariance matrix. $\sim$ stands for "distributed as" and $\cup$ is the union operator. 
\end{comment}
\begin{figure*}[t!]

\centering
\includegraphics[width=.28\textwidth]{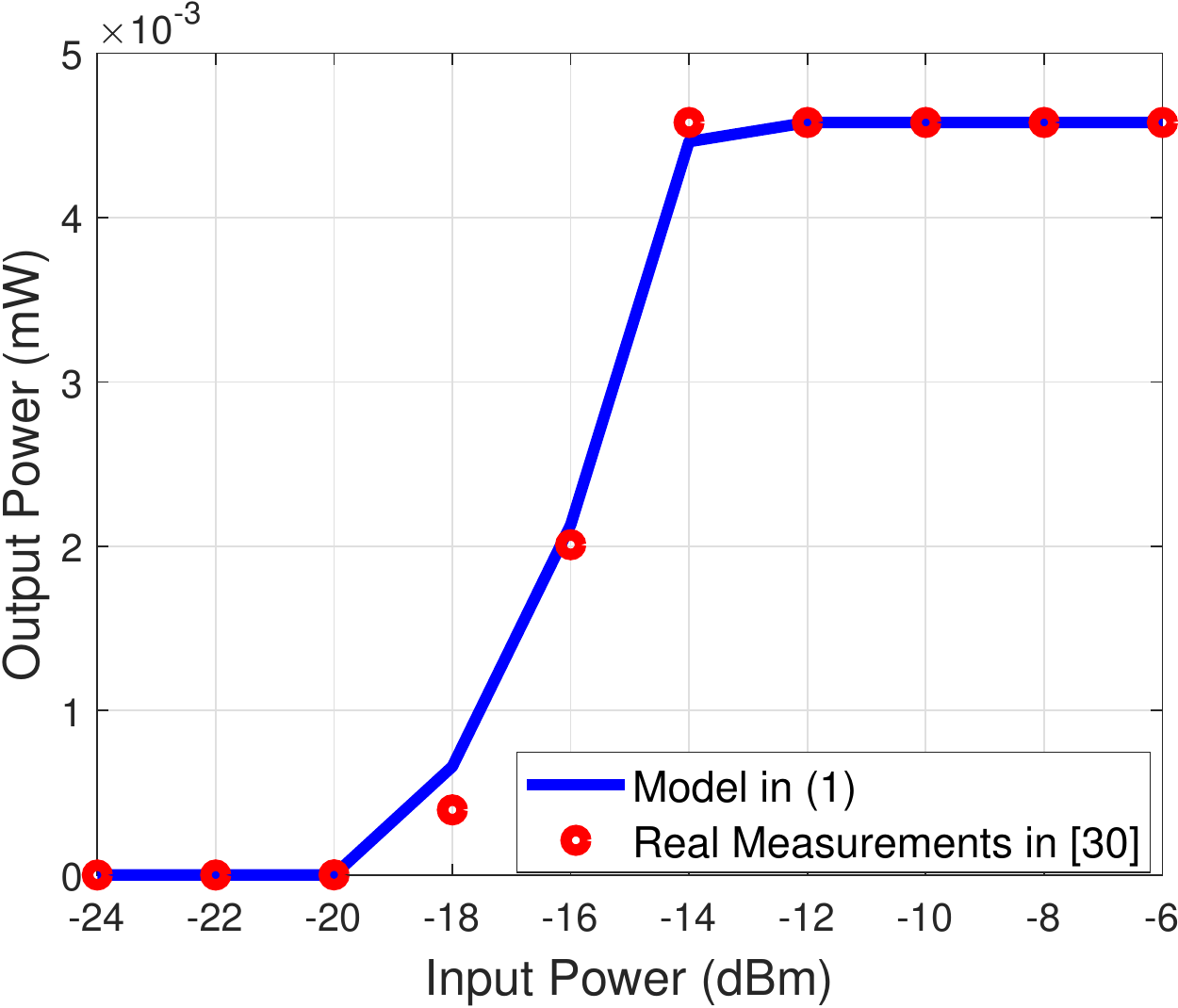}\hfill
\includegraphics[width=.28\textwidth]{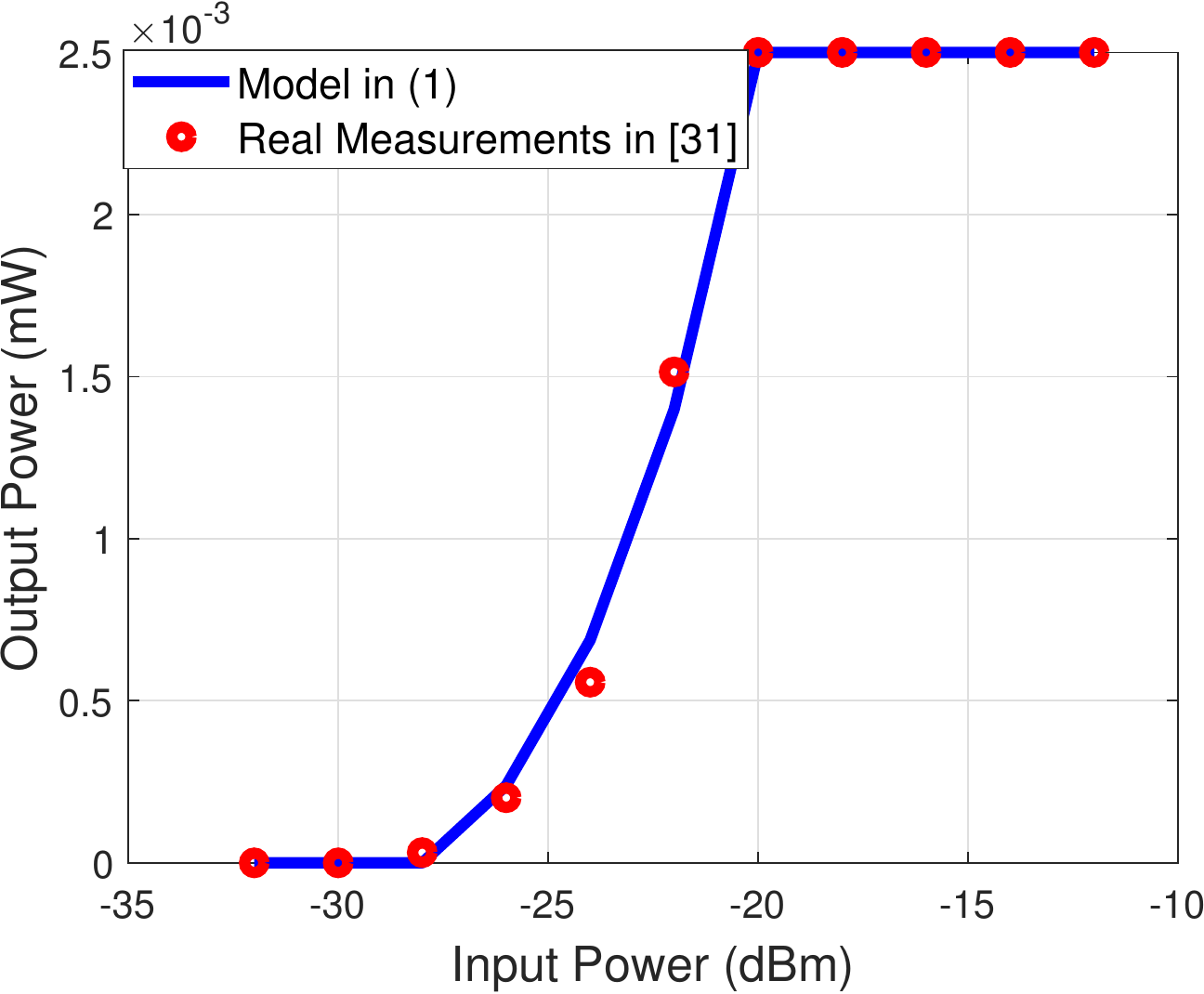}\hfill
\includegraphics[width=.28\textwidth]{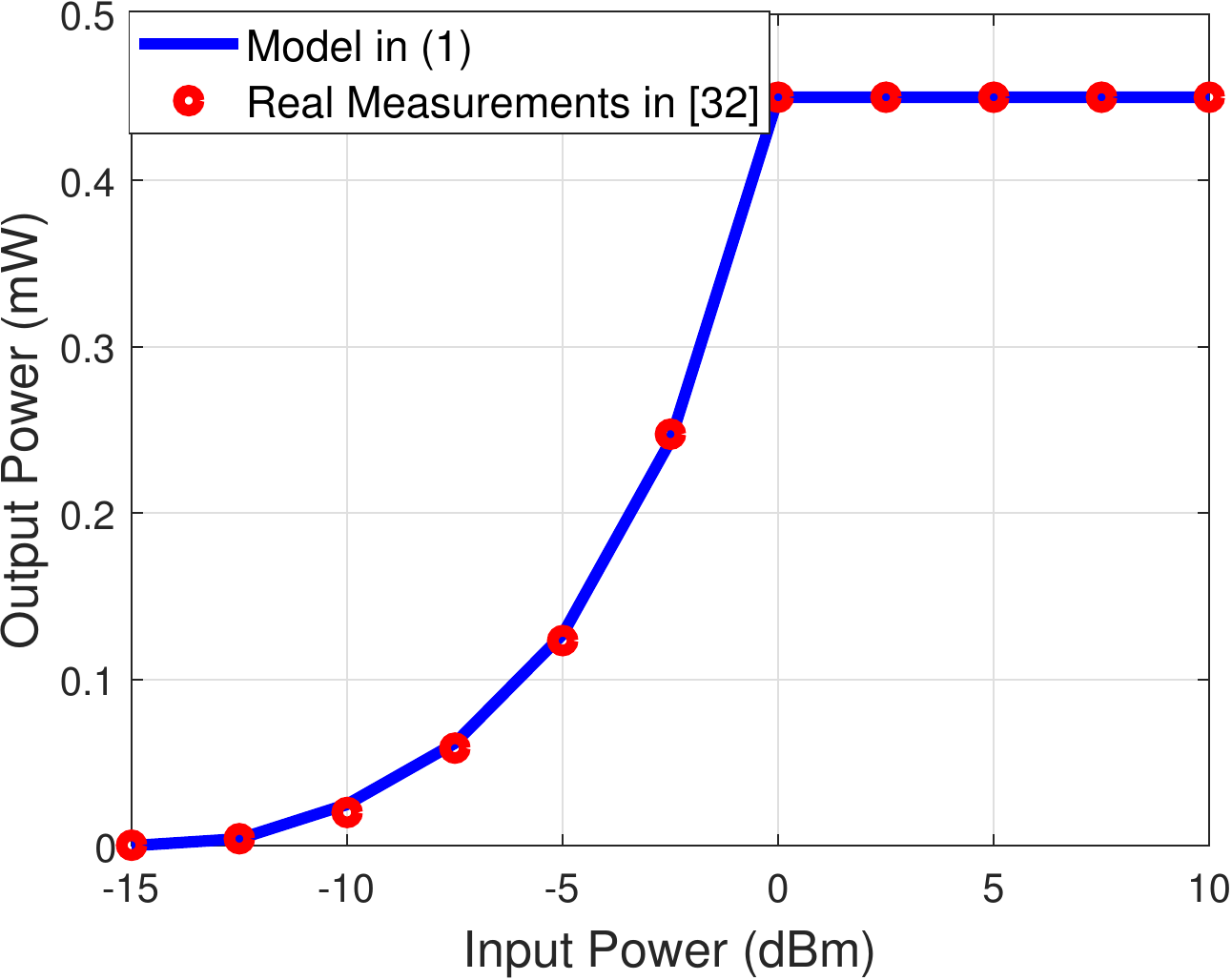}\hfill

\caption{Comparison between the model in (1) and real measurements in \cite{real_data1,real_data2,real_data3}}.
\label{curve}

\end{figure*}

\section{System Model and Problem Formulation}
We consider a multi-user BS-WPCN, as shown in Fig. \ref{fig1}, where the users are wirelessly powered by a PS to transmit their information to an AP. In addition to collecting energy from the PS transmissions, the users also backscatter the received signals from the PS to transmit information to the AP in a TDMA manner. The harvested energy will then be used by the users to actively transmit information signals to the AP. The backscatter and active IT of the users are aided by IRS elements which induce amplitude and phase changes to the signals transmitted by the users such that the signals from different paths are constructively added at the AP.

The PS and AP are equipped with $M_P$ and $M_A$ antennas, respectively, while each user has one single antenna. Also, the IRS consists of $N$ reflecting elements. We denote by $\mathcal{K}=\{1,...,K\}$ and $\mathcal{N}=\{1,...,N\}$ the set of all users and IRS reflecting elements, respectively. $\boldsymbol{F}^H\in \mathbb{C}^{K \times M_P}$, $\boldsymbol{G}\in \mathbb{C}^{N \times K}$, $\boldsymbol{C} \in \mathbb{C}^{M_A \times K}$, and $\boldsymbol{H} \in \mathbb{C}^{M_A \times N}$ respectively represent the PS-users, users-IRS, users-AP, and IRS-AP channel matrices. In this work, CSI is obtained based on the techniques proposed in \cite{irs-csi-1} 

In what follows, we present the system model including the EH model and the communication model, and formulate the sum-throughput maximization problem. 

\subsection{EH Model}
Most EH circuits use diode-based rectifiers to convert the received RF power into dc power. The efficiency of diode-based energy harvesters highly depends on the level of the received power. Specifically, when the received power is below the sensitivity of the energy harvester, the output power is zero because low input powers cannot turn on the diode. More importantly, when the received power exceeds some specific levels, the energy conversion efficiency is greatly degraded \cite{conv_eff1}; however, the output dc power remains constant and becomes saturated \cite{conv_eff2}. Taking into account the saturation effect is pivotal when studying networks with EH-enabled devices. Without considering this effect, the theoretical optimal design fails to perform properly in practice. Therefore, the conventional linear EH model cannot be relied on for studying the behavior of EH circuits. The sigmoidal EH model has been recently used by many researchers for designing resource allocation schemes in EH networks \cite{elena,bruno}. However, this model cannot be easily implemented in convex optimization toolboxes such as CVX \cite{cvx} without further approximations and transformations, especially when the received power varies over the harvesting period and the harvested power becomes a sum-of-ratios expression. Another EH model has been used in the literature which models the input-output power relationship of an energy harvester with a piece-wise linear function with two pieces, where the output power linearly increases with the input power up to the saturation point, beyond which the output power remains constant \cite{wpcn_relay4,bs2}. This model is more tractable than the sigmoidal model, however, it does not account for the sensitivity of EH circuits. 

In this paper, we use a piece-wise linear EH model with three pieces, where both sensitivity and saturation effects are taken into consideration. The harvested power in this model is given by 
\begin{align}
\label{our_pwl}
    p_h= \min\big(\max (0,\eta p_r - \xi), p_{\text{sat}}\big),
\end{align}where $p_h$ and $p_r$ are the harvested power and the received power, respectively, $p_{\text{sat}}$ is the saturation power and $\frac{\xi}{\eta}$ is the turn-on power (i.e., the sensitivity of the EH circuit). 

In Fig. \ref{curve}, we have plotted the curve fitting results for the piece-wise linear EH model in \eqref{our_pwl} using real measurements from \cite{real_data1,real_data2,real_data3}. The figure shows the good match between the model in \eqref{our_pwl} and real data, which confirms that this model is accurate enough for modeling the behavior of realistic EH circuits.  

\subsection{Communication Model}
The transmission block, which is depicted in Fig. \ref{fig2}, begins with PS acting as both energy and signal source for the users, transmitting unmodulated RF signals in the downlink.  Unlike the conventional WPCN, where the transmitted signals of the PS are merely used as a source of energy, the users of BS-WPCN make a more efficient use of these signals. In addition to harvesting their required energy, the users also modulate their information onto the signals and transmit them towards the AP using the backscattering technique. The backscatter transmission of the users is performed in a TDMA manner such that each user gets a dedicated time slot for performing backscatter transmission to the AP. Specifically, time slot $i$ of duration $\tau_i$ is allocated to backscatter transmission of the $i$-th user, denoted as $U_i$, while other users harvest energy from the signal transmitted by the PS. During $\tau_{K+1}$\footnote{Henceforth, time slots are represented by their duration, i.e., by $\tau_i$ we mean time slot $i$.}, all users simultaneously transmit to the AP, using their previously harvested energy. Both backscatter and active IT of the users are aided by IRS elements which apply phase and amplitude adjustments to the incident signals and reflect them to the AP. The real-time amplitude and phase adjustment of IRS elements are performed by a smart controller which is powered by a stable energy source (e.g., battery or grid). This model can be easily extended to the scenario where IRS elements are enabled to harvest their required energy from existing RF sources (e.g., AP, PS, etc.) by assuming that a portion of the transmission block (e.g., $\tau_0$) is dedicated for the EH of IRS.

During $\tau_i,~\forall i \in \mathcal{K}$, the PS transmits $\boldsymbol{w}_i \hat{s}_i$ in the downlink, where $\boldsymbol{w}_i$ is the beamforming vector of the PS and $\hat{s}_i$ is an unmodulated signal with unit power. The received signal at $U_i$ is given by $ \boldsymbol{f}_i^H \boldsymbol{w}_i \hat{s}_i$\footnote{We ignore noise at the users because no signal processing is performed.}, where $\boldsymbol{f}_i^H$ is the $i$-th row of $\boldsymbol{F}^H$. $U_i$ uses a portion of the received signal power for powering its circuit operations during backscatter transmission and reflects the remaining portion for transmitting information to the AP. Denote by $\beta_i$ the backscatter coefficient of $U_i$. The transmitted signal of $U_i$ during $\tau_i$ is given by $x_{1,i}=\sqrt{\beta_i}\boldsymbol{f}_i^H \boldsymbol{w}_i s_{1,i}\hat{s}_i$,    
where $s_{1,i}$ is the information-bearing signal of $U_i$ with $\mathbb{E}\{|s_{1,i}|^2\}=1$ and $\mathbb{E}[\cdot]$ shows the expected value. Each IRS element applies amplitude and phase changes to the received signal and reflects it to the AP. The received signal at the AP during $\tau_i,~\forall i \in \mathcal{K}$ is obtained as 

\begin{figure}[t!]
\centering
\includegraphics[width=2.8in]{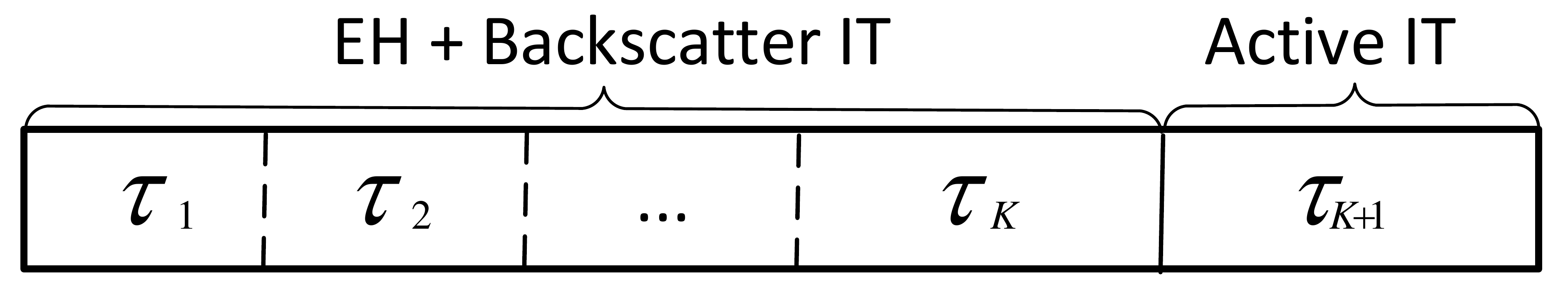}
\caption{ Transmission block structure for the IRS-empowered BS-WPCN}
	\label{fig2}
	
\end{figure}
 
\begin{align}
\label{y1}
    \boldsymbol{y}_{1,i}=(\boldsymbol{H}\boldsymbol{\Theta}_i\boldsymbol{g}_i+\boldsymbol{c}_i)x_{1,i}+\boldsymbol{n},
\end{align}where $\boldsymbol{g}_i$ and $\boldsymbol{c}_i$ are the $i$-th columns of $\boldsymbol{G}$ and $\boldsymbol{C}$, respectively. $\boldsymbol{\Theta}_i$ is the diagonal reflection matrix of the IRS during $\tau_i$ and its $n$-th diagonal element is given by $\alpha_{n,i}\text{exp}(j\theta_{n,i})$ with $\alpha_{n,i}$ and $\theta_{n,i}$ respectively representing the amplitude reflection coefficient and phase shift applied by the $n$-th element of IRS during $\tau_i$, and $\boldsymbol{n} \sim \mathcal{CN}(0,\sigma^2\boldsymbol{I}_{M_A})$ is the additive white Gaussian noise (AWGN) at the AP, where $\mathcal{CN}(\boldsymbol{\nu} ,\boldsymbol{\Sigma})$  is a circularly symmetric complex Gaussian (CSCG) random vector with $\boldsymbol{\nu}$ being the mean vector and $\boldsymbol{\Sigma}$ indicating the covariance matrix.

\begin{comment}
\[
  \boldsymbol{\Theta}_i =
  \begin{bmatrix}
    \alpha_{1,i} e^{j\theta_{1,i}} & & \\
    & \ddots & \\
    & & \alpha_{N,i} e^{j\theta_{N,i}}
  \end{bmatrix}
\]with $\alpha_{n,i}$ and $\theta_{n,i}$ respectively representing the amplitude reflection coefficient and phase shift applied by the $n$-th element of IRS during $\tau_i$.  
\end{comment}
From \eqref{y1} and after maximum ratio combining (MRC) at the AP, the signal-to-noise ratio (SNR) of $U_i$ is given by 
\begin{align}
    \gamma_{1,i}= \dfrac{\beta_i|\boldsymbol{f}_i^H\boldsymbol{w}_i|^2 ||\boldsymbol{h}_i(\boldsymbol{\Theta}_i)||^2}{\sigma^2},
\end{align}with $\boldsymbol{h}_i(\boldsymbol{\Theta}_i)=\boldsymbol{H}\boldsymbol{\Theta}_i\boldsymbol{g}_i+\boldsymbol{c}_i$ being the effective channel between AP and $U_i$ during $\tau_i$. The achievable throughput of $U_i$ during $\tau_i$ is obtained as 
\begin{align}
\label{R1i}
    R_{1,i}=\tau_i \log (1+\gamma_{1,i})=\tau_i \log (1+\dfrac{\beta_i|\boldsymbol{f}_i^H\boldsymbol{w}_i|^2 ||\boldsymbol{h}_i(\boldsymbol{\Theta}_i)||^2}{\sigma^2}).
\end{align}

While $U_i$ backscatters the transmitted signal of PS for information transfer, other users harvest energy from this signal. The harvested energy of $U_j~j\neq i$ during $\tau_i$ is given by 
\begin{align}
    e_{h,j,i}=\tau_i p_{h,j,i},
\end{align} where $p_{h,j,i}$ is the harvested power of $U_j$ during $\tau_i$ and is given by\footnote{Without loss of generality, we assume that the EH circuits of all users are similar such that the same EH parameters are used for all users. } 
\begin{align}
    p_{h,j,i}= \min\big(\max (0,\eta |\boldsymbol{f}_j^H \boldsymbol{w}_i|^2 - \xi), p_{\text{sat}}\big).
\end{align}

The users utilize their harvested energy for actively transmitting information to the AP.  Denote by $s_{2,i}$ the information signal of $U_i$ in $\tau_{K+1}$ and by $p_i$ the transmit power of $U_i$. The transmitted signal of $U_i$ during $\tau_{K+1}$ is then given by $x_{2,i}=\sqrt{p_i}s_{2,i}$. The transmitted signal of the users is again reflected by the IRS which applies the reflection matrix $\boldsymbol{\Theta}_{K+1}$ to the signal and transmits it to the AP. The received signal at the AP is given by 
\begin{align}
\boldsymbol{y}_2=\sum_{i=1}^{K} \boldsymbol{h}_i (\boldsymbol{\Theta}_{K+1}) x_{2,i}+\boldsymbol{n},    
\end{align}where $\boldsymbol{h}_i (\boldsymbol{\Theta}_{K+1})=\boldsymbol{H}\boldsymbol{\Theta}_{K+1}\boldsymbol{g}_i+\boldsymbol{c}_i$ and $\boldsymbol{\Theta}_{K+1}$ is defined similar to $\boldsymbol{\Theta}_i$ with $\alpha_{n,i}$ and $\theta_{n,i}$ being replaced by $\alpha_{n,K+1}$ and $\theta_{n,K+1},~\forall n \in \mathcal{N}$. 

The AP applies the receive beamforming vector $\boldsymbol{a}_{i}$ for decoding the signal of $U_i$. The SINR and achieveable throughput of $U_i$ for active IT are respectively obtained as  

\begin{align}
\label{gamma2}
    \gamma_{2,i}&=\dfrac{p_{i}|\boldsymbol{a}_{i}^H \boldsymbol{h}_i(\boldsymbol{\Theta}_{K+1})|^2}{\sum_{j \neq i} p_{j}|\boldsymbol{a}_{i}^H \boldsymbol{h}_j(\boldsymbol{\Theta}_{K+1})|^2 + ||\boldsymbol{a}_{i}^H||^2\sigma^2},\\
  \label{R2i} R_{2,i}&=\tau_{K+1}\log(1+\gamma_{2,i})\notag \\&=\tau_{K+1}\log\Big(1+\dfrac{p_{i}|\boldsymbol{a}_{i}^H \boldsymbol{h}_i(\boldsymbol{\Theta}_{K+1})|^2}{\sum_{j \neq i} p_{j}|\boldsymbol{a}_{i}^H \boldsymbol{h}_j(\boldsymbol{\Theta}_{K+1})|^2 + ||\boldsymbol{a}_{i}^H||^2\sigma^2}\Big).
\end{align}

\subsection{Problem Formulation}
We aim to maximize the total throughput of the network by optimizing IRS reflection coefficients in all time slots, the transmit beamforming at the PS, the receive beamforming at the AP, the power allocation of the users for active IT, and the time allocation of the network\footnote{The AP serves as a central node which has the required resources in terms of energy and processing capacity for implementing the algorithms and performing all the necessary computations. The optimized values of the variables are sent to IRS, users, and PS via dedicated control links. There also exist control links from IRS to the users and the PS. If the link between the AP and the users or between the AP and the PS fails, the control information can be exchanged through the IRS.}. The sum-throughput maximization problem is thus formulated as 
\begin{align}
    \label{P}
    \max_{\substack{\{\boldsymbol{\Theta}_{i}\}_{i=1}^K,\boldsymbol{\Theta}_{K+1}\\ \{\boldsymbol{w}_i\}_{i=1}^K, \{\boldsymbol{a}_i\}_{i=1}^K,\boldsymbol{p},\boldsymbol{\tau} }}~&\sum_{i=1}^K R_{1,i} + R_{2,i}\\
\text{s.t.}~\label{Avg}&\sum_{i=1}^{K}\tau_i\text{Tr}(\boldsymbol{w}_i \boldsymbol{w}_i^H)\leq p_{\text{avg}}, \tag{10.a}\\
\label{peak}& \text{Tr}(\boldsymbol{w}_i \boldsymbol{w}_i^H) \leq p_{\text{peak}},~~\forall i \in \mathcal{K}, \tag{10.b}\\
\label{EC}&(p_i+p_{c,i}) \tau_{K+1} \leq \sum_{j\neq i} e_{h,i,j},~\forall i \in \mathcal{K}, \tag{10.c}\\
\label{totaltime}&\sum_{i=1}^{K}\tau_i + \tau_{K+1} \leq 1, \tag{10.d}\\
\label{range}&p_i\geq 0,\forall i \in \mathcal{K}, ~\tau_{i} \geq 0,~\forall i \in \mathcal{K} \cup \{K+1\}, \tag{10.e}\\
\label{thetarange} &0  < \alpha_{n,i} \leq 1,~  0<\theta_{n,i}\leq 2\pi\notag \\ &~~~~~~~~~~~~~~~~~\forall n \in \mathcal{N},~i \in \mathcal{K} \cup \{K+1\}, \tag{10.f} 
\end{align}where $\boldsymbol{p}=[p_1,...,p_K]$ and $\boldsymbol{\tau}=[\tau_1,...,\tau_{K+1}]$. Also, $\cup$ is the union operator. Constraints \eqref{Avg} and \eqref{peak} account for the average power constraint and peak power constraint at the PS, respectively. Constraint \eqref{EC} is the energy causality constraint at the users with $p_{c,i}$ denoting the circuit power consumption at $U_i$ for active information transfer. Constraint \eqref{totaltime} is the total time constraint assuming a normalized transmission block. Constraints \eqref{range} and \eqref{thetarange} define the acceptable range for the optimization variables. The problem in \eqref{P} is non-convex and challenging to solve. In the following, we propose a two-stage scheme, where the IRS reflection coefficients for backscatter transmission are optimized in the first stage and the optimization of other variables is performed in the second stage. Note that the schemes proposed in this paper can also be used to solve the weighted sum-throughput maximization problem, where the weights can be arbitrarily set to indicate user priorities and to control user fairness.  

 \section{IRS Reflection for Backscatter Communication}

 \label{Ph1}
 The IRS reflection optimization problem for users' backscatter transmission is given by 
 \begin{align}
    \label{PIRS1}
    \max_{\substack{\{\boldsymbol{\Theta}_{i}\}_{i=1}^K}}~&\sum_{i=1}^K \tau_i \log (1+\hat{\beta_i} ||\boldsymbol{h}_i(\boldsymbol{\Theta}_i)||^2)\\
\text{s.t.}&~~ 0 < \alpha_{n,i} \leq 1,~  0<\theta_{n,i}\leq 2\pi, ~\label{alphthet}\forall n \in \mathcal{N},~i \in \mathcal{K}, \tag{11.a}
\end{align}where $\hat{\beta_i}=\frac{\beta_i |\boldsymbol{f}_i^H\boldsymbol{w}_i|^2}{\sigma^2}$.

We can split the problem in \eqref{PIRS1} into $K$ separate sub-problems, each dealing with the optimization of the IRS reflection coefficients in one of the time slots. Specifically, the $i$-th sub-problem will be
 \begin{align}
    \label{Phase1}
    \max_{\substack{\boldsymbol{\Theta}_{i}}}~&||\boldsymbol{h}_i(\boldsymbol{\Theta}_i)||^2=||\boldsymbol{H}\boldsymbol{\Theta}_i\boldsymbol{g}_i+\boldsymbol{c}_i||^2 \\
    \text{s.t.}~&~~ 0 < \alpha_{n,i} \leq 1,~  0<\theta_{n,i}\leq 2\pi,~\forall n \in \mathcal{N}. \tag{12.a}
 \end{align}
 
 Defining $\boldsymbol{V}_i=\boldsymbol{H}\text{diag}(\boldsymbol{g}_i)$ and $\boldsymbol{u}_i=[u_{1,i},...,u_{N,i}]^T$ with $u_{n,i}=\alpha_{n,i}\exp(j\theta_{n,i}),~\forall n \in \mathcal{N}$, we have \begin{align}
     ||\boldsymbol{H}\boldsymbol{\Theta}_i\boldsymbol{g}_i+\boldsymbol{c}_i||^2 &=(\boldsymbol{V}_i \boldsymbol{u}_i + \boldsymbol{c}_i)^H (\boldsymbol{V}_i \boldsymbol{u}_i + \boldsymbol{c}_i)
     \notag \\&=\boldsymbol{u}_i^H \boldsymbol{V}_i^H \boldsymbol{V}_i \boldsymbol{u}_i+2\text{Re}\{\boldsymbol{u}_i^H \boldsymbol{V}_i^H \boldsymbol{c}_i\} +||\boldsymbol{c}_i||^2,
 \end{align} and problem \eqref{Phase1} can be re-written as 
  \begin{align}
\label{Phase1p}
    \max_{\substack{\boldsymbol{u}_{i}}}~&\boldsymbol{u}_i^H \boldsymbol{V}_i^H \boldsymbol{V}_i \boldsymbol{u}_i+2\text{Re}\{\boldsymbol{u}_i^H \boldsymbol{V}_i^H \boldsymbol{c}_i\} +||\boldsymbol{c}_i||^2 \\
    \text{s.t.}~\label{uin}& |u_{n,i}| \leq 1,~\forall n \in \mathcal{N}. \tag{14.a}
 \end{align}
 
 Problem \eqref{Phase1p} is a non-convex optimization problem because of the quadratic term  $\boldsymbol{u}_i^H \boldsymbol{V}_i^H \boldsymbol{V}_i \boldsymbol{u}_i$ in the objective function. In the following, we propose two methods for solving problem \eqref{Phase1p}. In the first method, we alternately optimize the reflection coefficient of one IRS element having other reflection coefficients fixed, in an iterative manner. The second method is based on the SA technique, where a lower-bound of the objective function is maximized iteratively until convergence is achieved. 
\subsection{AO-Based Design} 
 Expanding the first two terms of the objective function in \eqref{Phase1p}, we have 
 \begin{small}\begin{align}
 \label{expansion}
     &\boldsymbol{u}_i^H\boldsymbol{V}_i^H \boldsymbol{V}_i \boldsymbol{u}_i+2\text{Re}\{\boldsymbol{u}_i^H \boldsymbol{V}_i^H\boldsymbol{c}_i\}=\sum_{n=1}^N \bigg(|u_{n,i}|^2  \sum_{m=1}^{M_A}|[\boldsymbol{V}_i]_{m,n}|^2+ \notag \\ &2\text{Re}\Big\{u_{n,i}\Big(\sum_{q=n+1}^N\overline{u}_{q,i}(\sum_{m=1}^{M_A} [\boldsymbol{V}_i]_{m,n}[\boldsymbol{V}^H_i]_{q,m})+ \sum_{m=1}^{M_A}\overline{c}_{m,i} [\boldsymbol{V}_i]_{m,n}\Big)  \Big\} \bigg),
 \end{align}\end{small}where $\overline{x}$ indicates the conjugate of $x$, $[\boldsymbol{V}_i]_{x,y}$ is the element on the $x$-th row and $y$-th column, and $c_{m,i}$ is the $m$-th element of $\boldsymbol{c}_i$.
 
The objective is to alternately optimize the reflection coefficients. The optimization problem for the $n$-th ($\forall n \in \mathcal{N}$) reflection coefficient in the $i$-th ($\forall i \in \mathcal{K}$) time slot is thus formulated as 

\begin{align}
\label{opt_uni}
   \max_{\substack{u_{n,i}}}~&t_{n,i}|u_{n,i}|^2+2\text{Re}\{z_{n,i}u_{n,i}\}\\
    \text{s.t.}~&|u_{n,i}| \leq 1, \notag  
\end{align}where
\begin{align}
    &t_{n,i}=\sum_{m=1}^{M_A}|[\boldsymbol{V}_i]_{m,n}|^2\notag\\&z_{n,i}= \sum_{\substack{q=1 \\ q\neq n}}^N \overline{u}_{q,i}(\sum_{m=1}^{M_A} [\boldsymbol{V}_i]_{m,n}[\boldsymbol{V}^H_i]_{q,m})+ \sum_{m=1}^{M_A}\overline{c}_{m,i} [\boldsymbol{V}_i]_{m,n},\notag
\end{align}and the terms independent of $u_{n,i}$ have been discarded. The optimal solution to \eqref{opt_uni} is readily obtained as 
\begin{align}
\label{ao-based}
    u_{n,i}=e^{-j\text{arg}(z_{n,i})}.
\end{align}
 
 Algorithm \ref{alg_zero} describes the above alternating procedure for optimizing the IRS reflection coefficients.
 \begin{center}
\begin{algorithm}
\SetAlgoLined
{\textbf{Inputs:} $\boldsymbol{H},\boldsymbol{g}_i,\boldsymbol{c}_i,\forall i \in \mathcal{K}$\;}
{\textbf{Outputs:} $\boldsymbol{\Theta}_i,~\forall i \in \mathcal{K}$\;}
{Set $u_{n,i}^{(0)}=1,\forall n \in \mathcal{N},i \in \mathcal{K}$\;}
 \For{i=1:K}{
  {$\Delta=1,l=0$\;}
  \While{$\Delta > \epsilon$}{
  {$l=l+1$\;}
  \For {n=1:N}{
  Given $u_{n^{\prime},i}^{(l)}$ for $n^{\prime}=1,...,n-1$ and $u_{n^{\prime}}^{(l-1)}$ for $n^{\prime}=n+1,...,N$, update $z_{n,i}$ \;
  Find $u_{n,i}^{(l)}$ from \eqref{ao-based}\;}
  $\Delta=||\boldsymbol{u}_i^{(l)}-\boldsymbol{u}_i^{(l-1)}||$\;}
   Set $\alpha_{n,i}^*=|u_{n,i}^{(l)}|$ and $\theta_{n,i}^*=\text{arg}(u_{n,i}^{(l)}),~ \forall n \in \mathcal{N}$\;
  }
 \caption{\begin{small}AO-Based Optimization of IRS Reflection Coefficients for Backscatter IT\end{small}} \label{alg_zero}
\end{algorithm}\end{center}

\subsection{SA-Based Design}

 As mentioned earlier in this section, the non-convexity of problem \eqref{Phase1p}  is due to the objective function being quadratic, which is convex in $\boldsymbol{u}_i$. At any feasible point $\boldsymbol{u}_i^{(0)}$, the term  $\boldsymbol{u}_i^H \boldsymbol{V}_i^H \boldsymbol{V}_i \boldsymbol{u}_i$ is lower-bounded by its first-order Taylor expansion as $2\text{Re}\{\boldsymbol{u}_i^H \boldsymbol{V}_i^H \boldsymbol{V}_i \boldsymbol{u}_i^{(0)}\} - \boldsymbol{u}_i^{(0)H} \boldsymbol{V}_i^H \boldsymbol{V}_i \boldsymbol{u}_i^{(0)}$. We can therefore apply the SA technique and iteratively maximize the lower-bound of the objective function in \eqref{Phase1p} until convergence. In iteration $l$, we will have the following optimization problem:
 \begin{align}
    \label{Phase1_SCA}
    \max_{\substack{\boldsymbol{u}_{i}}}~&2\text{Re}\{\boldsymbol{u}_i^H (\boldsymbol{V}_i^H \boldsymbol{V}_i \boldsymbol{u}_i^{(l-1)} + \boldsymbol{V}_i^H \boldsymbol{c}_i)\} + C_i^{(l-1)} \\
    \text{s.t.}~&\eqref{uin} \notag,
 \end{align}where $\boldsymbol{u}_i^{(l-1)}$ is the optimized $\boldsymbol{u}_i$ in the $(l-1)$-th iteration and $C_i^{(l-1)}=||\boldsymbol{c}_i||^2 - \boldsymbol{u}_i^{(l-1)H} \boldsymbol{V}_i^H \boldsymbol{V}_i \boldsymbol{u}_i^{(l-1)} $. Problem \eqref{Phase1_SCA} is equivalent to  
 \begin{align}
    \label{Phase1_Final}
    \max_{\substack{\boldsymbol{u}_i}}~&\text{Re}\big\{\sum_{n=1}^N 
     \overline{u}_{n,i}\Tilde{v}_{n,i}\big\} \\
    \text{s.t.}~&\eqref{uin} \notag,
 \end{align} where  $\Tilde{v}_{n,i}$ is the $n$-th element of $\Tilde{\boldsymbol{v}}_i=\boldsymbol{V}_i^H \boldsymbol{V}_i \boldsymbol{u}_i^{(l-1)} + \boldsymbol{V}_i^H \boldsymbol{c}_i$. It is straightforward to see that the optimal solution to problem \eqref{Phase1_Final} is given by
\begin{align}
\label{ref_phaseI}
    u_{n,i}=e^{j\text{arg}(\Tilde{v}_{n,i})}, ~\forall n \in \mathcal{N}.
\end{align}

The steps for optimizing the IRS reflection coefficients using the SA-based method is given in Algorithm \ref{alg_phase1}.

 \begin{center}\begin{algorithm}
 \SetAlgoLined
 {\textbf{Inputs:} $\boldsymbol{H},\boldsymbol{g}_i,\boldsymbol{c}_i,\forall i \in \mathcal{K}$\;}
{\textbf{Outputs:} $\boldsymbol{\Theta}_i,~\forall i \in \mathcal{K}$\;}
{Set $u_{n,i}^{(0)}=1,\forall n \in \mathcal{N},i \in \mathcal{K}$\;}
  \For{i=1:K}{
  {$\Delta=1,l=0$\;}
  \While{$\Delta > \epsilon$}{
  {$l=l+1$\;}
  {Calculate $\Tilde{\boldsymbol{v}}_i=\boldsymbol{V}_i^H \boldsymbol{V}_i \boldsymbol{u}_i^{(l-1)} + \boldsymbol{V}_i^H \boldsymbol{c}_i$\;}
  {Find $u_{n,i}^{(l)},\forall n \in \mathcal{N}$ from \eqref{ref_phaseI}\;}
  {$\Delta=||\boldsymbol{u}_i^{(l)}-\boldsymbol{u}_i^{(l-1)}||$\;}}
  { Set $\alpha_{n,i}^*=|u_{n,i}^{(l)}|$ and $\theta_{n,i}^*=\text{arg}(u_{n,i}^{(l)}),~ \forall n \in \mathcal{N}$.}
 }
  
\caption{\begin{small}SA-Based Optimization of IRS Reflection Coefficients for Backscatter IT\end{small}}
\label{alg_phase1}
\end{algorithm}\end{center}

\begin{remark}
\label{rem1}
 In Algorithm \ref{alg_zero},  although the optimization of $u_{n,i}$ in \eqref{opt_uni} ensures that the power of the combined signal from the direct path and the $n$-th reflected path is improved, it does not necessarily mean that the power of the collective signal received at the AP is also increased. In other words, the separate optimization of IRS reflection coefficients may fail to result in coherent combination of the individually reflected signals at the AP and the improvement of throughput over iterations cannot be ensured. The SA-based technique proposed in Algorithm \ref{alg_phase1} simultaneously updates the reflection coefficients of all IRS elements and thus, guarantees that the throughput increases after each update and a near-optimal solution for IRS reflection coefficients is obtained. Therefore, the SA-based method is expected to outperform the AO-based counterpart as will be shown in Section \ref{sims}. However, the complexity of Algorithm \ref{alg_zero} is lower than that of Algorithm \ref{alg_phase1} with  $\mathcal{O}(KNL_1)$ and $\mathcal{O}(KN^2 L_2)$ being the complexity orders of the AO-based and SA-based techniques, respectively, where $L_1$ and $L_2$ are the number of iterations needed for the convergence of the corresponding while loops. 
\end{remark}

As stated in Remark \ref{rem1}, the objective function is guaranteed to improve after each iteration with Algorithm \ref{alg_phase1}. Specifically, setting $\mathcal{F}(\boldsymbol{u}_i)=\boldsymbol{u}_i^H \boldsymbol{V}_i^H \boldsymbol{V}_i \boldsymbol{u}_i+2\text{Re}\{\boldsymbol{u}_i^H \boldsymbol{V}_i^H \boldsymbol{c}_i\} +||\boldsymbol{c}_i||^2$, we have

\begin{align}
\label{conv_proof}
 &\mathcal{F}(\boldsymbol{u}_i^{(l)})\overset{(\varpi_1)}{\geq} 2\text{Re}\{\boldsymbol{u}_i^{(l)H} (\boldsymbol{V}_i^H \boldsymbol{V}_i \boldsymbol{u}_i^{(l-1)} + \boldsymbol{V}_i^H \boldsymbol{c}_i)\} + C_i^{(l-1)}\notag\\&\overset{(\varpi_2)}{\geq} 2\text{Re}\{\boldsymbol{u}_i^{(l-1)H} (\boldsymbol{V}_i^H \boldsymbol{V}_i \boldsymbol{u}_i^{(l-1)} + \boldsymbol{V}_i^H \boldsymbol{c}_i)\} + C_i^{(l-1)}=\mathcal{F}(\boldsymbol{u}_i^{(l-1)})  \notag\\
 &\overset{(\varpi_3)}{\geq} 2\text{Re}\{\boldsymbol{u}_i^{(l-1)H} (\boldsymbol{V}_i^H \boldsymbol{V}_i \boldsymbol{u}_i^{(l-2)} + \boldsymbol{V}_i^H \boldsymbol{c}_i)\} + C_i^{(l-2)},~~\forall l,
\end{align}where $(\varpi_1)$ and $(\varpi_3)$ hold because $\mathcal{F}(\boldsymbol{u}_i)$ is lower-bounded by its first-order Taylor expansion, and $(\varpi_2)$ is due to the fact that $\boldsymbol{u}_i^{(l)}$ is the optimal solution to \eqref{Phase1_SCA} in the $l$-th iteration. Therefore, both the objective function in \eqref{Phase1p} and its lower-bound in \eqref{Phase1_SCA} increase after each iteration.

\section{Resource Allocation, Beamforming, and IRS Reflection for Active Communication}
In order to find the near-optimal design for the remaining optimization variables, the AO technique is used by dividing the variables into groups and alternately optimizing them in an iterative manner. The details will be elaborated in the following subsections.
\subsection{Design of Resource Allocation and Transmit Beamforming}
 We first investigate the optimization of resource allocation including power allocation at the users and time allocation for EH, backscatter IT, and active IT, as well as the transmit beamforming vectors of the PS. We have the following optimization problem: 

 \begin{align}
    \label{resource}
    \max_{\substack{\boldsymbol{p},\boldsymbol{\tau},\{\boldsymbol{w}_i\}_{i=1}^K}} ~&\sum_{i=1}^K \bigg(\tau_i \log \Big(1+\dfrac{\beta_i ||\boldsymbol{h}_i(\boldsymbol{\Theta}_i)||^2 |\boldsymbol{f}_i^H\boldsymbol{w}_i|^2 }{\sigma^2}\Big)\notag\\&+\tau_{K+1}\log \Big(1+\dfrac{p_{i}|\boldsymbol{a}_i^H \boldsymbol{h}_i(\boldsymbol{\Theta}_{K+1})|^2}{\sum_{j \neq i} p_{j}|\boldsymbol{a}_i^H \boldsymbol{h}_j(\boldsymbol{\Theta}_{K+1})|^2 + ||\boldsymbol{a}_i^H||^2\sigma^2}\Big)\bigg)\\
\text{s.t.}&~~~\eqref{Avg}-\eqref{range}. \notag
\end{align}

The above problem is not a convex optimization problem because the objective function is not concave and the variables are coupled in the objective function and the constraints. We define  $\Tilde{\boldsymbol{W}}_i = \tau_i \boldsymbol{w}_i \boldsymbol{w}_i^H$, $e_{i}= \tau_{K+1} p_{i}$ and introduce auxiliary matrix $\boldsymbol{\Phi}$. Problem \eqref{resource} is re-written as

\begin{align}
    \label{resource2}
    \max_{\substack{\boldsymbol{e},\boldsymbol{\tau}},\{\Tilde{\boldsymbol{W}}_i\}_{i=1}^K,\boldsymbol{\Phi}} ~&\sum_{i=1}^K \bigg(\tau_i \log \Big(1+\Tilde{\beta}_i\dfrac{\text{Tr}(\Tilde{\boldsymbol{F}}_i \Tilde{\boldsymbol{W}}_i) }{\tau_i}\Big)\notag\\&+\tau_{K+1}\log\Big(1+\dfrac{\Tilde{a}_{i,i} \dfrac{e_{i}}{\tau_{K+1}}}{\sum_{j \neq i} \Tilde{a}_{i,j}\dfrac{e_{j}}{\tau_{K+1}}+ \Tilde{c}_i}\Big)\bigg)\\
\text{s.t.}~\label{avg2}&\sum_{i=1}^{K}\text{Tr}(\Tilde{\boldsymbol{W}}_i)\leq p_{\text{avg}} \tag{23.a}\\
\label{peak2}& \text{Tr}(\Tilde{\boldsymbol{W}}_i) \leq \tau_i p_{\text{peak}},~~\forall i \in \mathcal{K} \tag{23.b}\\
\label{EC2}&e_i+p_{c,i} \tau_{K+1} \leq \sum_{j\neq i} [\boldsymbol{\Phi}]_{i,j},~\forall i \in \mathcal{K}\tag{23.c}\\
&[\boldsymbol{\Phi}]_{i,j}+\xi \tau_{j} \leq \eta \text{Tr}(\Tilde{\boldsymbol{F}}_i \Tilde{\boldsymbol{W}}_j), \forall i, j \in \mathcal{K},~j \neq i, \tag{23.d} \\
&[\boldsymbol{\Phi}]_{i,j} \leq \tau_j p_{\text{sat}}, \forall i,j \in \mathcal{K},~j \neq i, \tag{23.e}\\
\label{totaltime2}&\sum_{i=1}^{K}\tau_i + \tau_{K+1} \leq 1, \tag{23.f}\\
\label{range2}&e_i \geq 0, \forall i \in \mathcal{K},~~\tau_{i}\geq 0, \forall i \in \mathcal{K} \cup \{K+1\}, \tag{23.g}\\
\label{Wpos}&\Tilde{\boldsymbol{W}}_i \geq 0,~\forall i \in \mathcal{K}, \tag{23.h}\\
\label{rank}& \text{Rank}(\Tilde{\boldsymbol{W}}_i)= 1,~\forall i \in \mathcal{K}, \tag{23.i}
\end{align}where $\boldsymbol{e}=[e_1,...,e_K]$, $\Tilde{\boldsymbol{F}}_i=\boldsymbol{f}_i \boldsymbol{f}_i^H$, $\Tilde{\beta}_i=\frac{\beta_i ||\boldsymbol{h}_i(\boldsymbol{\Theta}_i)||^2}{\sigma^2}$, $\Tilde{a}_{i,j}=|\boldsymbol{a}_i^H\boldsymbol{h}_j(\boldsymbol{\Theta_{K+1}})|^2$, and $\Tilde{c}_i=||\boldsymbol{a}_i^H||^2 \sigma^2$. 

Problem \eqref{resource2} is still non-convex because the second term of the objective function is not concave and also the rank-one constraint in \eqref{rank} is not convex. To deal with the non-concavity in the objective function, we write the second term of \eqref{resource2} as a sum of concave and convex functions and apply the SA technique to iteratively maximize a lower bound of the objective function. Specifically, we have
\begin{align}
\label{concave-convex}
   &\tau_{K+1}\log(1+\dfrac{\Tilde{a}_{i,i} \dfrac{e_{i}}{\tau_{K+1}}}{\sum_{j \neq i} \Tilde{a}_{i,j}\dfrac{e_{j}}{\tau_{K+1}}+ \Tilde{c}_i})=\notag\\&\tau_{K+1}\log\big(\sum_{j=1}^K \Tilde{a}_{i,j}\dfrac{e_j}{\tau_{K+1}}+\Tilde{c}_i \big)-
   \tau_{K+1}\log\big(\sum_{\substack{j=1 \\  j\neq i}}^K \Tilde{a}_{i,j}\dfrac{e_j}{\tau_{K+1}}+\Tilde{c}_i \big).
\end{align}

In \eqref{concave-convex}, $\tau_{K+1}\log\big(\sum_{j=1}^K \Tilde{a}_{i,j}(e_j/\tau_{K+1})+\Tilde{c}_i\big)$ is a concave function since it is obtained by applying the perspective operation to the concave function $\log\big(\sum_{j=1}^K \Tilde{a}_{i,j}e_j+\Tilde{c}_i\big)$ and the perspective operation preserves concavity. On the other hand, $-\tau_{K+1}\log\big(\sum_{j\neq i} \Tilde{a}_{i,j}(e_j/\tau_{K+1})+\Tilde{c}_i\big)$ is a jointly convex function of $\tau_{K+1}$ and $\boldsymbol{e}$, which motivates us to use the SA technique for solving \eqref{resource2}. Based on the first-order Taylor series expansion, this convex function can be approximated by its lower bound as

\begin{small}\begin{align}
    &-\tau_{K+1}\log\big(\sum_{j\neq i} \Tilde{a}_{i,j}\dfrac{e_j}{\tau_{K+1}}+\Tilde{c}_i \big) \approx -\tau_{K+1}^{(0)}\log\big(\sum_{j\neq i} \Tilde{a}_{i,j}\dfrac{{e}_j^{(0)}}{{\tau_{K+1}^{(0)}}}+\Tilde{c}_i \big)\notag \\ & +\Big(\dfrac{\sum_{j \neq i} \Tilde{a}_{i,j}\dfrac{e^{(0)}_{j}}{\tau^{(0)}_{K+1}}}{\sum_{j \neq i} \Tilde{a}_{i,j}\dfrac{e^{(0)}_{j}}{\tau^{(0)}_{K+1}}+\Tilde{c}_i} - \log \big(\sum_{j \neq i} \Tilde{a}_{i,j}\dfrac{e^{(0)}_{j}}{\tau^{(0)}_{K+1}}+\Tilde{c}_i\big)\Big)(\tau_{K+1}-\tau^{(0)}_{K+1})\notag\\&+ \Big(-\dfrac{1}{\sum_{j \neq i} \Tilde{a}_{i,j}\dfrac{e^{(0)}_{j}}{\tau^{(0)}_{K+1}}+\Tilde{c}_i}\Big)\big(\sum_{j\neq i} \Tilde{a}_{i,j}(e_j-e^{(0)}_j)\big),
\end{align}\end{small}where $\tau^{(0)}_{K+1}$ and $e^{(0)}_j$ are feasible values for $\tau_{K+1}$ and $e_j$, respectively. Applying the SA method, the optimization problem in iteration $l$ can be formulated as follows:

%\noindent\makebox[\linewidth]{\rule{\textwidth}{1pt}} 
\begin{align}
\label{resource3}
    \max_{\substack{\boldsymbol{e},\boldsymbol{\tau},\{\Tilde{\boldsymbol{W}}_i\}_{i=1}^K}} ~&R_{\text{new}}^{(l)}\\
\text{s.t.}~&~\eqref{avg2}-\eqref{rank} \notag
\end{align}where $R_{\text{new}}^{(l)}$ is given in \eqref{Rnew} and superscript $(l-1)$ indicates the optimized value in iteration $(l-1)$. 
\begin{align}
\label{Rnew}
&R_{\text{new}}^{(l)}=\notag\\&\sum_{i=1}^K \bigg(\tau_i \log \big(1+\Tilde{\beta}_i\dfrac{\text{Tr}(\Tilde{\boldsymbol{F}}_i \Tilde{\boldsymbol{W}}_i) }{\tau_i}\big)+\tau_{K+1}\log\big(1+\sum_{j=1}^K \dfrac{\Tilde{a}_{i,j}}{\Tilde{c}_i}\dfrac{e_j}{\tau_{K+1}}\big)\notag\\&+\delta_{e,i}^{(l-1)}(\sum_{j \neq i}\Tilde{a}_{i,j}e_j)+\delta_{\tau,i}^{(l-1)}\tau_{K+1}\bigg),
\end{align}
\begin{small}\begin{align}
    &\delta_{e,i}^{(l-1)}=-\dfrac{1}{\sum_{j \neq i} \Tilde{a}_{i,j}\dfrac{{e}_{j}^{(l-1)}}{\tau_{K+1}^{(l-1)}}+\Tilde{c}_i},\notag\\
    &\delta_{\tau,i}^{(l-1)}=\dfrac{\sum_{j \neq i} \Tilde{a}_{i,j}\dfrac{e_{j}^{(l-1)}}{\tau_{K+1}^{(l-1)}}}{\sum_{j \neq i} \Tilde{a}_{i,j}\dfrac{{e}_{j}^{(l-1)}}{\tau_{K+1}^{(l-1)}}+\Tilde{c}_i} - \log \big(\sum_{j \neq i} \Tilde{a}_{i,j}\dfrac{e_{j}^{(l-1)}}{\tau_{K+1}^{(l-1)}}+\Tilde{c}_i\big)+\log(\Tilde{c}_i). \notag
\end{align}\end{small}

Now,  the only source of non-convexity for problem \eqref{resource3} is the rank-one constraint in \eqref{rank}, which can be relaxed using the semidefinite relaxation (SDR) technique. The relaxed problem will be a convex optimization problem, which can be solved by convex optimization toolboxes (e.g., CVX). 
 We iteratively maximize the lower bound of the total throughput by solving the relaxed version of \eqref{resource3} until a satisfactory convergence is achieved. In problems \eqref{resource2} and \eqref{resource3}, we have ignored the sensitivity of the EH circuits at the users. Therefore, the obtained value for some of $[\boldsymbol{\Phi}]_{i,j}$'s may be negative which is not acceptable. We thus have to perform one final step to find energy and time allocation. Setting $[\Tilde{\boldsymbol{\Phi}}]_{i,j}=\max (0,\boldsymbol{\Phi}_{i,j}),~\forall i,j \in \mathcal{K},~j\neq i$, we will have the following problem for optimizing energy and time:

\begin{align}
\label{resource4}
    \max_{\substack{\boldsymbol{e},\boldsymbol{\tau}}} ~&R_{\text{new}}^{(l)},\\
\text{s.t.}~&\eqref{totaltime2} \text{ and } \eqref{range2}, \notag \\
&\text{Tr}(\Tilde{\boldsymbol{W}}_i^*) \leq \tau_i p_{\text{peak}},~\forall i \in \mathcal{K}, \tag{28.a}\\
&~e_i + p_{c,i}\tau_{K+1} \leq \sum_{j \neq i} [\Tilde{\boldsymbol{\Phi}}]_{i,j},~\forall \in \mathcal{K}, \tag{28.b}
\end{align}where $\Tilde{\boldsymbol{W}}_i^*$ is the solution obtained for $\Tilde{\boldsymbol{W}}_i$ by solving problem \eqref{resource3}.  Problem \eqref{resource4} is a convex problem and can be solved either analytically or by CVX.  After iteratively solving problem \eqref{resource4}, the near-optimal energy and time allocation are obtained, which are denoted as $\boldsymbol{e}^*$ and $\boldsymbol{\tau}^*$. The near-optimal power allocation is then calculated as $p_i^*=e_i^*/\tau_{K+1}^*,~\forall i \in \mathcal{K}$.
 
Finally, if $\Tilde{\boldsymbol{W}}_i^{*}$ is not a rank-one matrix, we use the Eigen-decomposition technique to extract a feasible rank-one solution from it. Particularly, the rank-one approximation for $\Tilde{\boldsymbol{W}}_i^{*}$ is given by $\hat{\boldsymbol{W}}_i=\lambda_{i,1} \hat{\boldsymbol{w}}_{i,1} \hat{\boldsymbol{w}}_{i,1}^H$, where $\lambda_{i,1}$ is the largest eigenvalue of $\Tilde{\boldsymbol{W}}_i^{*}$ and 
$\hat{\boldsymbol{w}}_{i,1}$ is the corresponding eigenvector. The sub-optimal beamforming vector of the PS in the $i$-th time slot will be obtained as $\boldsymbol{w}_i^{*}=\sqrt{\lambda_{i,1}/\tau_i^{*}}\hat{\boldsymbol{w}}_{i,1}$. Algorithm \ref{alg_resource} summarizes the procedure for optimizing the network resource allocation and transmit beamforming vectors of the PS. 
\begin{center}
\begin{algorithm}
 \SetAlgoLined
{\textbf{Inputs:} $\boldsymbol{H},\boldsymbol{g}_i,\boldsymbol{c}_i,\boldsymbol{f}_i,\boldsymbol{a}_i, p_{c,i},\forall i \in \mathcal{K},\sigma^2,p_{\text{avg}},p_{\text{peak}},p_{\text{sat}},\eta,\xi$\;}
{\textbf{Outputs:} $\boldsymbol{p},\boldsymbol{\tau}, \boldsymbol{w}_i,~\forall i \in \mathcal{K}$\;}
{$\Delta=1,l=0, R_{\text{new}}^{(0)}=0$\;}
{Initialize $\boldsymbol{e}^{(0)}$ and $\boldsymbol{\tau}^{(0)}$\;}
 \While{$\Delta > \epsilon$}{
    {$l=l+1$\;}
         {Update $\delta_{e,i}^{(l-1)}$ and $\delta_{\tau,i}^{(l-1)}, \forall i \in \mathcal{K}$\;}
         {Solve \eqref{resource3} using CVX\;}
         { $\Delta=|R_{\text{new}}^{(l)}-R_{\text{new}}^{(l-1)}|$\;}}
          {Set $\Tilde{\boldsymbol{W}}_i^*=\Tilde{\boldsymbol{W}}_i^{(l)}, \forall i \in \mathcal{K}$\;}
          {Set $[\Tilde{\boldsymbol{\Phi}}]_{i,j}=\max (0,[\boldsymbol{\Phi}]_{i,j}),~\forall i,j \in \mathcal{K},~j\neq i$\;}
         {$\Delta=1,l=0, R_{\text{new}}^{(0)}=0$\;}
{Initialize $\boldsymbol{e}^{(0)}$ and $\boldsymbol{\tau}^{(0)}$\;}
 \While{$\Delta > \epsilon$}{
    {$l=l+1$\;}
         {Update $\delta_{e,i}^{(l-1)}$ and $\delta_{\tau,i}^{(l-1)}, \forall i \in \mathcal{K}$\;}
         {Solve \eqref{resource4}\;}
         { $\Delta=|R_{\text{new}}^{(l)}-R_{\text{new}}^{(l-1)}|$\;}}
 {Use Eigen-decomposition technique to obtain a feasible rank-one  $\hat{\boldsymbol{W}}_i$ from $\Tilde{\boldsymbol{W}}_i^{*},~ \forall i \in \mathcal{K}$\;}
  {Set $e_i^{*}=e_i^{(l)},\forall i\in \mathcal{K}, \tau_i^{*}=\tau_{i}^{(l)},\forall i \in \mathcal{K}, \tau_{K+1}^*=\tau_{K+1}^{(l)}, p_i^*=\frac{e_i^*}{\tau_{K+1}^*}, \forall i \in \mathcal{K}, \boldsymbol{w}_i^*=\sqrt{\lambda_{i,1}/\tau_i^*} \hat{\boldsymbol{w}}_{i,1},\forall i \in \mathcal{K}$\;}
\caption{\begin{small}Optimization of Resource Allocation and Transmit Beamforming \end{small}}
\label{alg_resource}
\end{algorithm}\end{center}
\subsection{Design of IRS Reflection for Active IT and Receive Beamforming}
 Having fixed the receive beamforming at the AP, the problem for optimizing the IRS reflection during $\tau_{K+1}$ is formulated as 
 \begin{align}
    \label{PhaseII}
    \max_{\substack{\boldsymbol{\Theta}_{K+1}}}~&\sum_{i=1}^K \bigg(\log\Big(1+\dfrac{p_{i}|\boldsymbol{a}_{i}^H \boldsymbol{h}_i(\boldsymbol{\Theta}_{K+1})|^2}{\sum_{j \neq i} p_{j}|\boldsymbol{a}_{i}^H \boldsymbol{h}_j(\boldsymbol{\Theta}_{K+1})|^2 + ||\boldsymbol{a}_{i}^H||^2\sigma^2}\Big)\bigg)\\
\text{s.t.}&~~\label{ThetaK+1} 0 < \alpha_{n,K+1} \leq 1,~  0 <\theta_{n,K+1}\leq 2\pi,~~\forall n \in \mathcal{N}. \tag{29.a}
\end{align}

Problem \eqref{PhaseII} is not a convex optimization problem and finding its optimal solution is not straightforward. A sub-optimal solution may be obtained by using the techniques of SDR and SA as will be briefly discussed in the following.

Using the previously defined $\boldsymbol{V}_i,~\forall i \in \mathcal{K}$ and setting $\boldsymbol{u}_{K+1}=[u_{1,K+1},...,u_{N,K+1}]^T$ with $u_{n,K+1}=\alpha_{n,K+1}\exp(j\theta_{n,K+1})$, $\Tilde{\boldsymbol{u}}_{K+1}=[\boldsymbol{u}_{K+1}^T 1]^T$, and  $\Tilde{\boldsymbol{U}}_{K+1}=\Tilde{\boldsymbol{u}}_{K+1}\Tilde{\boldsymbol{u}}_{K+1}^H$, and after dropping the rank-one constraint on $\Tilde{\boldsymbol{U}}_{K+1}$, problem \eqref{PhaseII}
is re-formulated as 
\begin{align}
    \label{PhaseII_2}
    \max_{\substack{\Tilde{\boldsymbol{U}}_{K+1}}}~&\sum_{i=1}^K \bigg(\log\Big(1+\dfrac{\text{Tr}(\boldsymbol{A}_{i,i}\Tilde{\boldsymbol{U}}_{K+1})}{\text{Tr}(\Tilde{\boldsymbol{A}}_i\Tilde{\boldsymbol{U}}_{K+1}) + ||\boldsymbol{a}_{i}^H||^2\sigma^2}\Big)\bigg)\\
\text{s.t.}&~\label{ubar1}[\Tilde{\boldsymbol{U}}_{K+1}]_{n,n} \leq 1, \forall n \in \mathcal{N},  \tag{30.a}\\&\label{ubar2}[\Tilde{\boldsymbol{U}}_{K+1}]_{N+1,N+1}=1 \tag{30.b},
\end{align}where
\begin{align}
    \boldsymbol{A}_{i,j}=p_j\begin{bmatrix}
    \boldsymbol{V}_j^H \boldsymbol{a}_i\boldsymbol{a}_i^H\boldsymbol{V}_j^H & \boldsymbol{V}_j^H \boldsymbol{a}_i \boldsymbol{a}_i^H\boldsymbol{c}_j  \\
    \boldsymbol{c}_j^H\boldsymbol{a}_i\boldsymbol{a}_i^H\boldsymbol{V}_j& 0
  \end{bmatrix},~\forall i,j \in \mathcal{K}, \notag
\end{align}and $\Tilde{\boldsymbol{A}}_i=\sum_{j \neq i}\boldsymbol{A}_{i,j}$. 

Clearly, problem \eqref{PhaseII_2} is still non-convex. Based on the product rule for logarithms, \eqref{PhaseII_2} is re-written as    
\begin{align}
    \label{log_rewrite}
    \max_{\substack{\Tilde{\boldsymbol{U}}_{K+1},\varsigma_i}}~&\sum_{i=1}^K \bigg(\log\Big(\text{Tr}(\boldsymbol{A}_{i,i}\Tilde{\boldsymbol{U}}_{K+1})+\text{Tr}(\Tilde{\boldsymbol{A}}_i\Tilde{\boldsymbol{U}}_{K+1}) + ||\boldsymbol{a}_{i}^H||^2\sigma^2\Big)\notag\\&-\log \varsigma_i  \bigg)\\
\text{s.t.}&~\eqref{ubar1} \text{ and }\eqref{ubar2} \notag,\\&\text{Tr}(\Tilde{\boldsymbol{A}}_i\Tilde{\boldsymbol{U}}_{K+1}) + ||\boldsymbol{a}_{i}^H||^2\sigma^2 \leq \varsigma_i,~\forall i \in \mathcal{K}, \tag{31.a}
\end{align}

Problem \eqref{log_rewrite} is still not convex because the objective function involves summation of concave and convex functions. This issue can be dealt with through the SA technique. Specifically, the convex term $-\log \varsigma_i$ can be approximated by its lower bound (i.e., first-order Taylor expansion). The resulting problem will be an SDP which can be iteratively solved via CVX, updating $\varsigma_i$ in each iteration until convergence is attained. Eventually, the rank-one extraction procedure must be performed to find a feasible rank-one solution for $\Tilde{\boldsymbol{U}}_{K+1}$, from which a sub-optimal $\boldsymbol{u}_{K+1}$ is obtained. A similar process as above can be applied for optimizing the receive beamformers at the AP \big($\{\boldsymbol{a}_i\}_{i=1}^K$\big), fixing other optimization variables.  

The above procedure for optimizing $\boldsymbol{u}_{K+1}$ and $\{\boldsymbol{a}_i\}_{i=1}^K$ incurs considerable complexity as SDP problems must be solved several times for each sub-problem. This complexity increases with increasing the number of users, antennas at the AP, and IRS elements, and can be prohibitively high in large networks. 

Herein, we propose an algorithm with lower complexity for optimizing IRS reflection coefficients in time slot $K+1$, assuming minimum mean square error (MMSE) receive beamforming at the AP. Specifically, using the relationship between mean square error (MSE) and SINR in MMSE receivers, we can jointly optimize the MMSE receive beamforming vectors at the AP and the IRS reflection coefficients during $\tau_{K+1}$, using the BCD technique. 

Under the assumption of independence between different $s_{2,i}$'s and also between $s_{2,i}, \forall i \in \mathcal{K}$ and each element of the noise vector $\boldsymbol{n}$, the MSE for $U_i$'s information signal is given by  

\begin{align}
\label{MSE}
   E_i&= \mathbb{E}[|\boldsymbol{a}_i^H \boldsymbol{y}_2 - s_{2,i}|^2] \notag\\&=\sum_{j=1}^K p_{j}|\boldsymbol{a}_i^H \boldsymbol{h}_j (\boldsymbol{\Theta}_{K+1})|^2 - \sqrt{p_{i}}\big(\boldsymbol{a}_i^H \boldsymbol{h}_i(\boldsymbol{\Theta}_{K+1}) +  \boldsymbol{h}_i^H(\boldsymbol{\Theta}_{K+1}) \boldsymbol{a}_i\big) \notag\\&+||\boldsymbol{a}_i^H||^2 \sigma^2+1.
\end{align}

The following theorem establishes an equivalence between throughput maximization and MSE minimization problems. 
\begin{theorem}
The problem in \eqref{PhaseII} is equivalent to the following problem
\begin{align}
    \label{equiv}
    \min_{\substack{ \{\boldsymbol{a}_i\}_{i=1}^K,\boldsymbol{\omega},\boldsymbol{\Theta}_{K+1}}}~&\sum_{i=1}^K \big( \omega_i E_i - \log \omega_i\big)\\
\text{s.t.}&~\eqref{ThetaK+1}, \notag
\end{align}where $\boldsymbol{\omega}=[\omega_1,...,\omega_K]$, and $\omega_i$ is a weight variable associated with $U_i$.

\end{theorem}
\begin{proof}
Please refer to the Appendix.
\end{proof}

Rewriting the MSE in \eqref{MSE} with respect to $\boldsymbol{V}_i, ~\forall i \in \mathcal{K}$ and $\boldsymbol{u}_{K+1}$, we have 

\begin{align}
\label{ei2}
    E_i&=\boldsymbol{u}_{K+1}^H \Big(\sum_{j=1}^K p_{j} \boldsymbol{V}_j^H\boldsymbol{a}_i \boldsymbol{a}_i^H \boldsymbol{V}_j\Big)\boldsymbol{u}_{K+1} \notag\\&+ \boldsymbol{u}_{K+1}^H\Big(\sum_{j=1}^K p_{j}\boldsymbol{V}_j^H \boldsymbol{a}_i \boldsymbol{a}_i^H \boldsymbol{c}_j- \sqrt{p_{i}}\boldsymbol{V}_i^H \boldsymbol{a}_i\Big)\notag\\&+ \Big(\sum_{j=1}^K p_{j}\boldsymbol{V}_j^H \boldsymbol{a}_i \boldsymbol{a}_i^H \boldsymbol{c}_j- \sqrt{p_{i}}\boldsymbol{V}_i^H \boldsymbol{a}_i\Big)^H \boldsymbol{u}_{K+1} \notag\\&+ \sum_{j=1}^K p_{j} \boldsymbol{c}_j^H \boldsymbol{a}_i\boldsymbol{a}_i^H \boldsymbol{c}_j  - \sqrt{p_{i}}(\boldsymbol{a}_i^H \boldsymbol{c}_i + \boldsymbol{c}_i^H \boldsymbol{a}_i) + ||\boldsymbol{a}_i^H||^2 \sigma^2 + 1.
\end{align}

According to Theorem 1, problem \eqref{PhaseII} can be re-written as 
\begin{align}
    \label{equiv3}
    \min_{\substack{ \{\boldsymbol{a}_i\}_{i=1}^K,\boldsymbol{\omega},\boldsymbol{u}_{K+1}}}~&\sum_{i=1}^K \big(\omega_i E_i - \log \omega_i \big)\\
\text{s.t.}&~\label{u_K1} |u_{n,K+1}| \leq 1,~\forall n \in \mathcal{N} \tag{35.a}, 
\end{align}

Applying the BCD method, problem \eqref{equiv3} can be alternately solved for $\{\boldsymbol{a}_i\}_{i=1}^K$, $\boldsymbol{\omega}$ and $\boldsymbol{u}_{K+1}$ in an iterative manner, where the optimum value for each variable is found in each iteration. Particularly, denoting  $\boldsymbol{u}_{K+1}^{(l-1)}$ as the optimal value for $\boldsymbol{u}_{K+1}$ in iteration $l-1$, the optimal receive beamforming vector for $U_i$ in iteration $l$ is given by
\begin{align}
\label{opt_a}
    &\boldsymbol{a}_i^{(l)}=\notag\\&\sqrt{p_{i}}\Bigg(\sum_{j=1}^K p_{j} \Big((\boldsymbol{V}_j \boldsymbol{u}_{K+1}^{(l-1)}+\boldsymbol{c}_j) (\boldsymbol{V}_j \boldsymbol{u}_{K+1}^{(l-1)}+\boldsymbol{c}_j)^H\Big) + \sigma^2 \boldsymbol{I}_{M_A}\Bigg)^{-1}\times \notag\\&(\boldsymbol{V}_i \boldsymbol{u}_{K+1}^{(l-1)}+\boldsymbol{c}_i),~\forall i \in \mathcal{K},
\end{align}and the optimal value for $\omega_i$ in iteration $l$ is given by
\begin{align}
\label{opt_omega}
\omega_i^{(l)}=\dfrac{1}{E_i ^{(l)}}, ~~\forall i \in \mathcal{K}, 
\end{align}where $E_i^{(l)}$ is obtained by substituting $\boldsymbol{u}_{K+1}^{(l-1)}$ and $\boldsymbol{a}_i^{(l)}$ into \eqref{ei2}. Finally, with $\{\boldsymbol{a}_i^{(l)}\}_{i=1}^K$ and $\boldsymbol{\omega}^{(l)}$, the problem for optimizing $\boldsymbol{u}_{K+1}$ in iteration $l$ is formulated as 
\begin{align}
    \label{PhaseII-Final}
    \min_{\substack{\boldsymbol{u}_{K+1}}}~& \boldsymbol{u}_{K+1}^H\boldsymbol{B}^{(l)} \boldsymbol{u}_{K+1} - 2 \text{Re} \{\boldsymbol{b}^{(l)H} \boldsymbol{u}_{K+1}\}\\
\text{s.t.}&~\label{QCQP_cons}\boldsymbol{u}_{K+1}^H \boldsymbol{D}_n \boldsymbol{u}_{K+1} \leq 1, ~\forall n \in \mathcal{N} \tag{38.a},
\end{align}where 
\begin{align}
    &\boldsymbol{B}^{(l)}= \sum_{i=1}^K \omega_i \sum_{j=1}^K p_{j} \boldsymbol{V}_j^H\boldsymbol{a}_i^{(l)} \boldsymbol{a}_i^{(l)H} \boldsymbol{V}_j\notag\\&\boldsymbol{b}^{(l)} = \sum_{i=1}^K \omega_i \Big(\sqrt{p_{i}}\boldsymbol{V}_i^H \boldsymbol{a}_i^{(l)}-\sum_{j=1}^K p_{j}\boldsymbol{V}_j^H \boldsymbol{a}_i^{(l)} \boldsymbol{a}_i^{(l)H} \boldsymbol{c}_j\Big), \notag 
\end{align}and $\boldsymbol{D}_n$ is a diagonal matrix with $1$ on its $n$-th diagonal element and $0$ elsewhere. \eqref{PhaseII-Final} is a convex quadratically-constrained quadratic program (QCQP) and can be solved using convex optimization techniques. We use the Lagrange duality method to solve \eqref{PhaseII-Final} for which the Lagrangian is given by 
\begin{align}
   & \mathcal{L}=\notag\\&\boldsymbol{u}_{K+1}^H\boldsymbol{B}^{(l)} \boldsymbol{u}_{K+1} - 2 \text{Re} \{\boldsymbol{b}^{(l)H} \boldsymbol{u}_{K+1}\} + \sum_{n=1}^N \mu_n (\boldsymbol{u}_{K+1}^H\boldsymbol{D}_n\boldsymbol{u}_{K+1}-1),
\end{align}with $\boldsymbol{\mu}=[\mu_1,...,\mu_N]$ being the Lagrange multipliers associated with the constraint \eqref{QCQP_cons}. The first-order optimality condition of $\mathcal{L}$ with respect to $\boldsymbol{u}_{K+1}$ yields
\begin{align}
    \boldsymbol{u}_{K+1} (\boldsymbol{\mu})= \big(\boldsymbol{B}^{(l)}+\sum_{n=1}^{N}\mu_n \boldsymbol{D}_n\big)^{-1}\boldsymbol{b}^{(l)}.
\end{align} $\boldsymbol{\mu}$ can be updated via the ellipsoid method with the subgradient of $\mu_n$ being given by 
   $|u_{n,K+1}|^2 - 1$. Algorithm \ref{alg_phase2} describes the steps for optimizing the IRS reflection coefficients for assisting users' active information transfer, when MMSE receiver is used at the AP. By alternately running algorithms \ref{alg_resource} and \ref{alg_phase2} until convergence, the near-optimal resource allocation, transmit and receive beamforming, and IRS reflection for active IT are obtained.
\begin{center}
\begin{algorithm}
 \SetAlgoLined
{\textbf{Inputs:} $\boldsymbol{H},\boldsymbol{g}_i,\boldsymbol{c}_i,\forall i \in \mathcal{K},\sigma^2, \boldsymbol{p}$\;}
{\textbf{Outputs:} $\boldsymbol{\Theta}_{K+1},\boldsymbol{a}_i,\forall i$\;}
{Initialize $\boldsymbol{\omega}^{(0)}$ and set $u_{n,K+1}^{(0)}=1, \forall n \in \mathcal{N}$\;}
  {$\Delta=1$, $l=0$\;}
  \While{$\Delta>\epsilon$}{
  {$l=l+1$\;}
  {Obtain $\boldsymbol{a}_i^{(l)},\forall i \in \mathcal{K}$ from \eqref{opt_a}\;}
  {Obtain $\omega_i^{(l)}, \forall i \in \mathcal{K}$ from \eqref{opt_omega}\;}
  {Obtain $\boldsymbol{u}_{K+1}^{(l)}$ by solving the QCQP in \eqref{PhaseII-Final}\;}
  {$\Delta=\big|\sum_{i=1}^K \log \big(\omega_i^{(l)}\big) - \sum_{i=1}^K \log \big(\omega_i^{(l-1)}\big)\big|$\;}}
  {Set $\alpha_{n,K+1}^*=|u_{n,K+1}^{(l)}|$,  $\theta_{n,K+1}^*=\text{arg}(u_{n,K+1}^{(l)}),~ \forall n \in \mathcal{N}$ and $\boldsymbol{a}_i^*=\boldsymbol{a}_i^{(l)},~\forall i \in \mathcal{K}$\;}
 \caption{\begin{small}Optimization of IRS Reflection for Active IT and Receive Beamforming \end{small}}
\label{alg_phase2}
\end{algorithm}\end{center}
\begin{remark}
The optimal IRS amplitude reflection coefficients in the active IT phase are not necessarily equal to 1. The reason is that the users simultaneously transmit to the AP and as is clear from the throughput expression of the users in the active IT phase (Eq. \eqref{R2i}), the power of the received signal from each user at the AP affects the throughput of other users. Therefore, optimizing the amplitude reflection of IRS elements in addition to their phase shifts is  important for maximizing the total throughput. This is in contrast to the backscatter transmission phase, where the optimal amplitude reflection coefficients are 1 (Eqs. \eqref{ao-based} and \eqref{ref_phaseI}) because each user individually transmits to the AP in its assigned time slot and maximizing the received power for each user’s signal maximizes its own throughput and the total throughput.
\end{remark}

\begin{remark}
Matrix $\boldsymbol{B}^{(l)}$ is the summation of $K^2$ rank-one matrices. If $K^2 \geq N$, matrix $\boldsymbol{B}^{(l)}$ is full-rank and invertible and so is $\boldsymbol{B}^{(l)}+\sum_{n=1}^N \mu_n \boldsymbol{D}_n$. If $K^2 < N$, the invertibility of $\boldsymbol{B}^{(l)}+\sum_{n=1}^N \mu_n \boldsymbol{D}_n$ depends on the value of $\mu_n$'s and is not guaranteed; but problem \eqref{PhaseII-Final} is still a QCQP and can be solved using CVX solvers. 
\end{remark}

\begin{remark}
The main complexity of the total throughput maximization algorithm belongs to Algorithm \ref{alg_resource}, in which SDP problems are solved via CVX. The complexity for solving each CVX problem in Algorithm \ref{alg_resource} is given by $\mathcal{O}\big(\max \{K^2,M_P\}^4 M_P^{0.5} \log (1/\varepsilon)\big)$, where $\varepsilon$ is the solution accuracy in CVX. Therefore, the overall complexity of the two-stage throughput maximization algorithm is approximated as $\mathcal{O}\big(L M_P^{4.5} \log (1/\varepsilon)\big)$ for $K^2 \leq M_P$ and  $\mathcal{O}\big(L K^8 M_P^{0.5} \log (1/\varepsilon)\big)$ for $K^2 > M_P$, where $L$ is the number of times CVX is called. 
\end{remark}

\section{Performance Evaluation}
\label{sims} 
\subsection{Simulation Setup}
We consider a 2-D Cartesian coordinate system, as shown in Fig. \ref{simset}, where the AP is located at the origin, the reference element of the IRS is placed at $(x_{\text{IRS}},y_{\text{IRS}})$ and the PS is positioned at $(x_{\text{PS}},0)$. $K$ users are evenly placed on the left half-circle centered at the PS with radius $r$.  Parameters $\eta$, $\xi$, and $p_{\text{sat}}$ for the EH model at the users are obtained by fitting the model in \eqref{our_pwl} to the real measurements reported in \cite{real_data3}. All channels are modeled by the Rician fading channel model \cite{irs3,irs4,irs10}. For example, the channel between the AP and the IRS is given by $\boldsymbol{H}=\sqrt{\frac{\kappa_h}{\kappa_h+1}}\boldsymbol{H}^{\text{LoS}}+\sqrt{\frac{1}{\kappa_h +1}}\boldsymbol{H}^{\text{NLoS}}$, where $\kappa_h$ is the Rician factor, $\boldsymbol{H}^{\text{LoS}}\in \mathbb{C}^{M_A \times N}$ and $\boldsymbol{H}^{\text{NLoS}} \in \mathbb{C}^{M_A \times N}$ are the line-of-sight (LoS) and non-line-of-sight (NLoS) components of $\boldsymbol{H}$. The LoS channel matrix is modeled as $\boldsymbol{H}^{\text{LoS}}=\hat{\boldsymbol{h}}_{M_A}(\varphi_{\text{AoA}})\hat{\boldsymbol{h}}_{N}^H(\varphi_{\text{AoD}})$,
where $\varphi_{\text{AoA}}$ and $\varphi_{\text{AoD}}$ denote the angle of arrival and angle of departure of IRS, respectively, and 
    $\hat{\boldsymbol{h}}_X(\varphi)=[1,e^{j\pi\sin(\varphi)},e^{j2\pi\sin(\varphi)},...,e^{j(X-1)\pi\sin(\varphi)}]^T$. The elements of $\boldsymbol{H}^{\text{NLoS}}$ follow the standard Rayleigh fading. $\boldsymbol{H}$ is then multiplied by the square root of the distance-dependent path-loss $C_0 (d_h/D_0)^{-\rho_h}$, where $C_0$ is the path-loss at the reference distance of $D_0=1$ meter (m), set as $C_0=-20$ dB, $d_h$ represents the distance between AP and IRS, and $\rho_h$ is the path-loss exponent of the channel between AP and IRS. Channels $\boldsymbol{F}$, $\boldsymbol{G}$, and $\boldsymbol{C}$ are defined in a similar way as $\boldsymbol{H}$, with $\kappa_f$, $\kappa_g$, and $\kappa_c$ being the corresponding Rician factors, and $\rho_f$, $\rho_g$, and $\rho_c$ denoting the path-loss exponents of the corresponding channels. 

\begin{figure}[t!]
\centering
\includegraphics[width=2.5in]{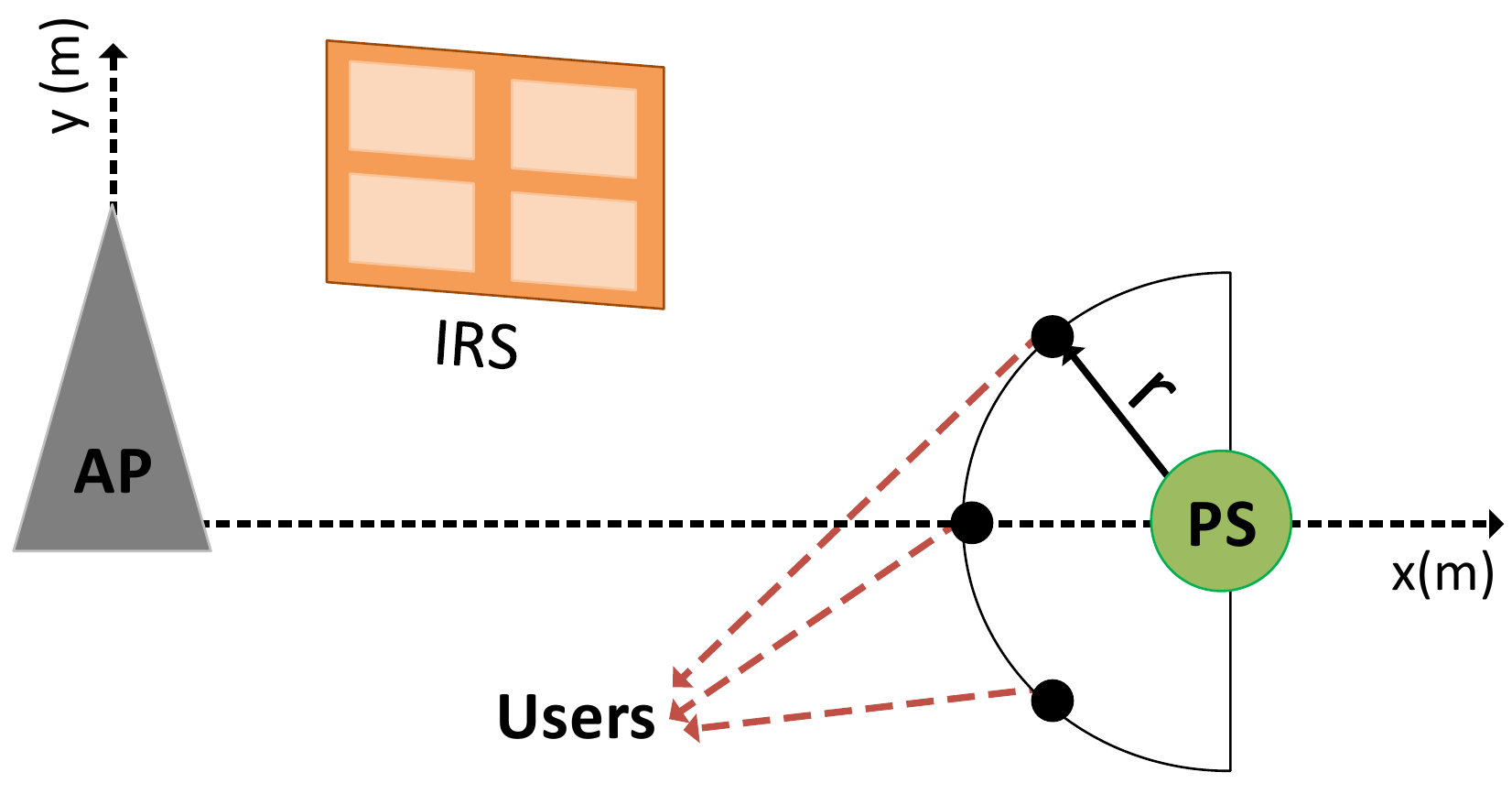}
\caption{ Simulation setup}
	\label{simset}	
\end{figure}

\begin{figure}[t!]
\centering
\includegraphics[width=1.8in]{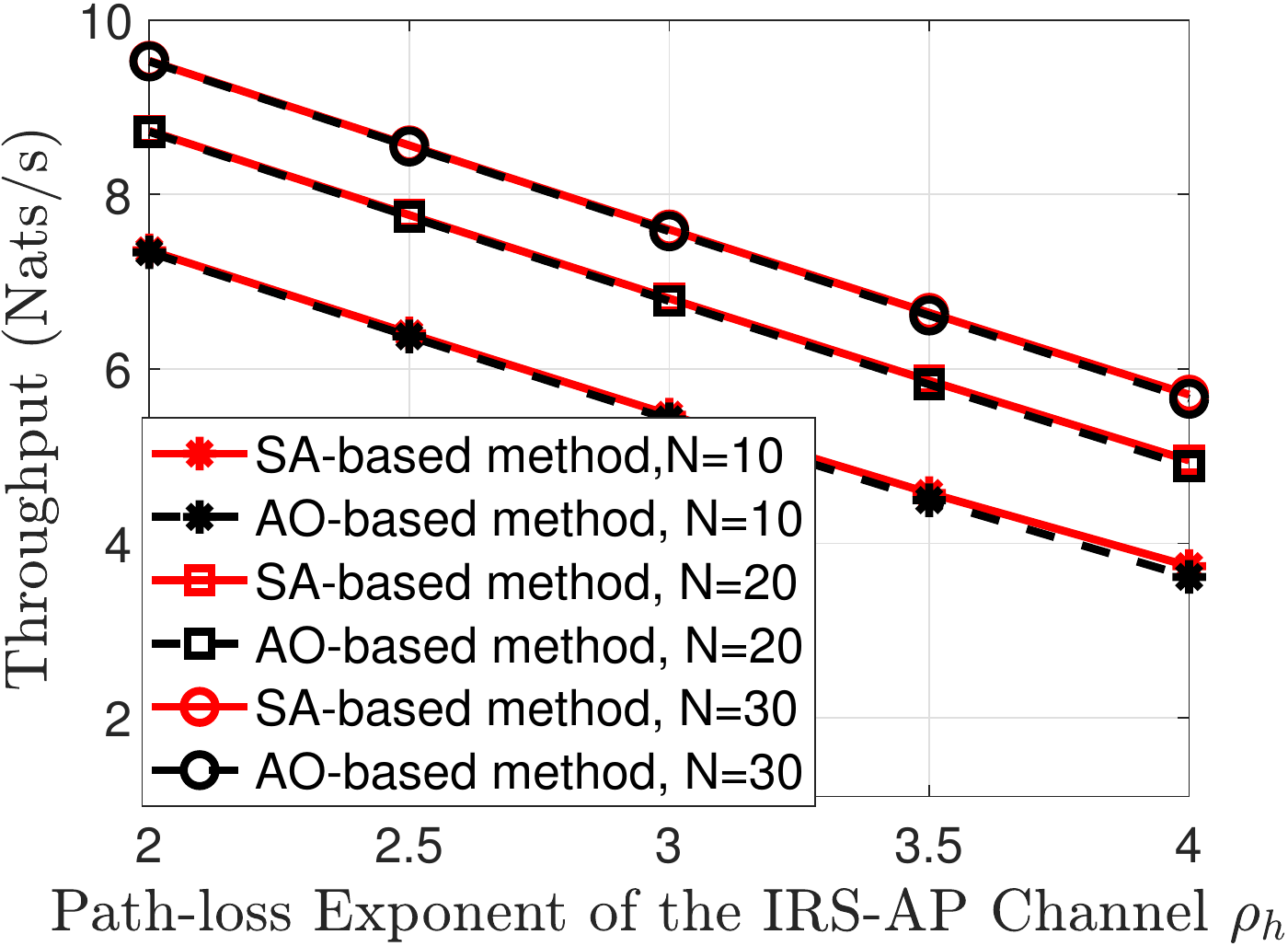}
\caption{ Comparison between SA-based and AO-based methods}
	\label{comp}
	
\end{figure}

\begin{figure}[t!]
\centering
%\def\tabularxcolumn#1{m{#1}}
%\begin{tabularx}{\linewidth}{@{}cXX@{}}
%\begin{tabular}{cc}
\subfloat[]{\includegraphics[width=1.8in]{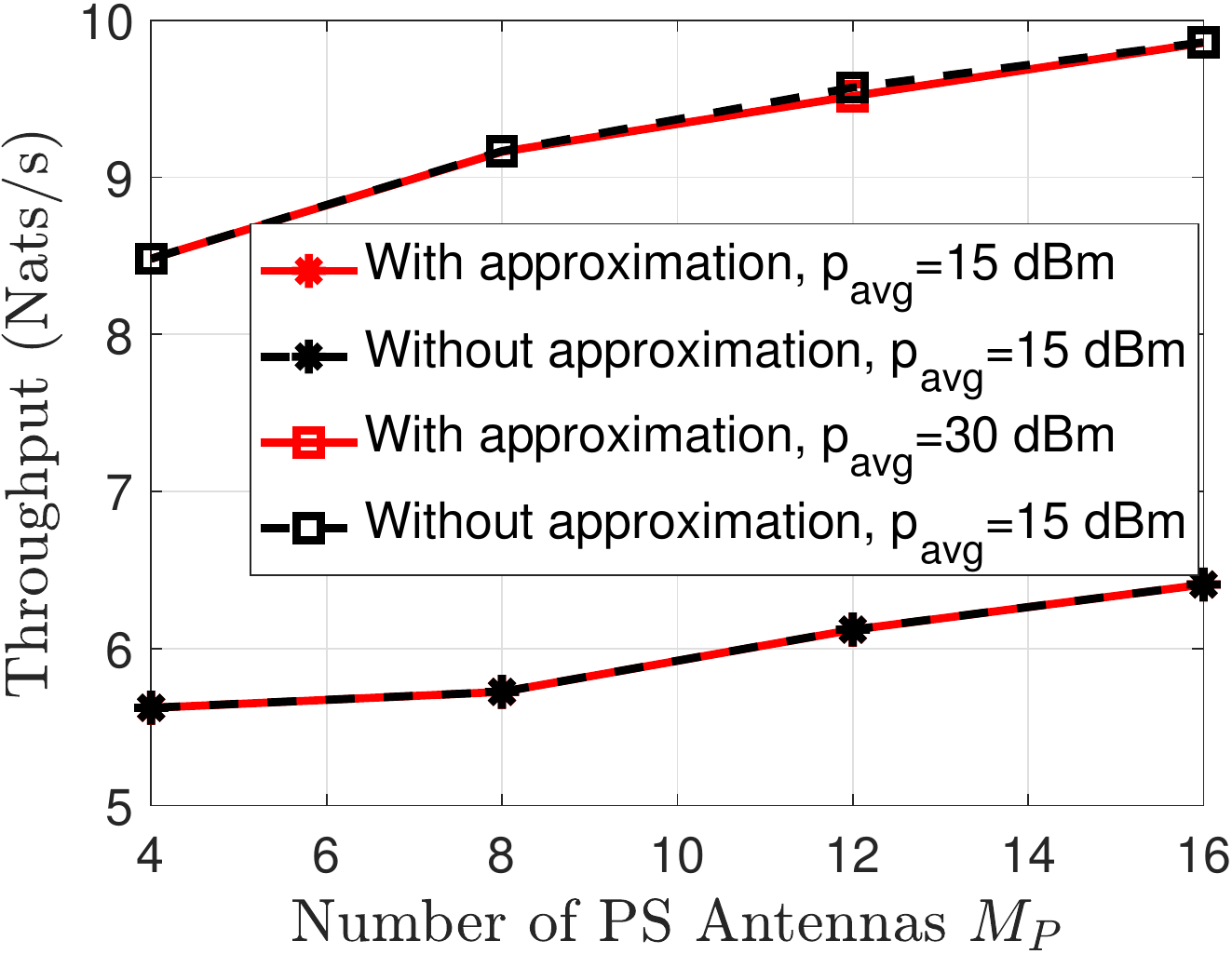}} \\
    \subfloat[]{\includegraphics[width=1.8in]{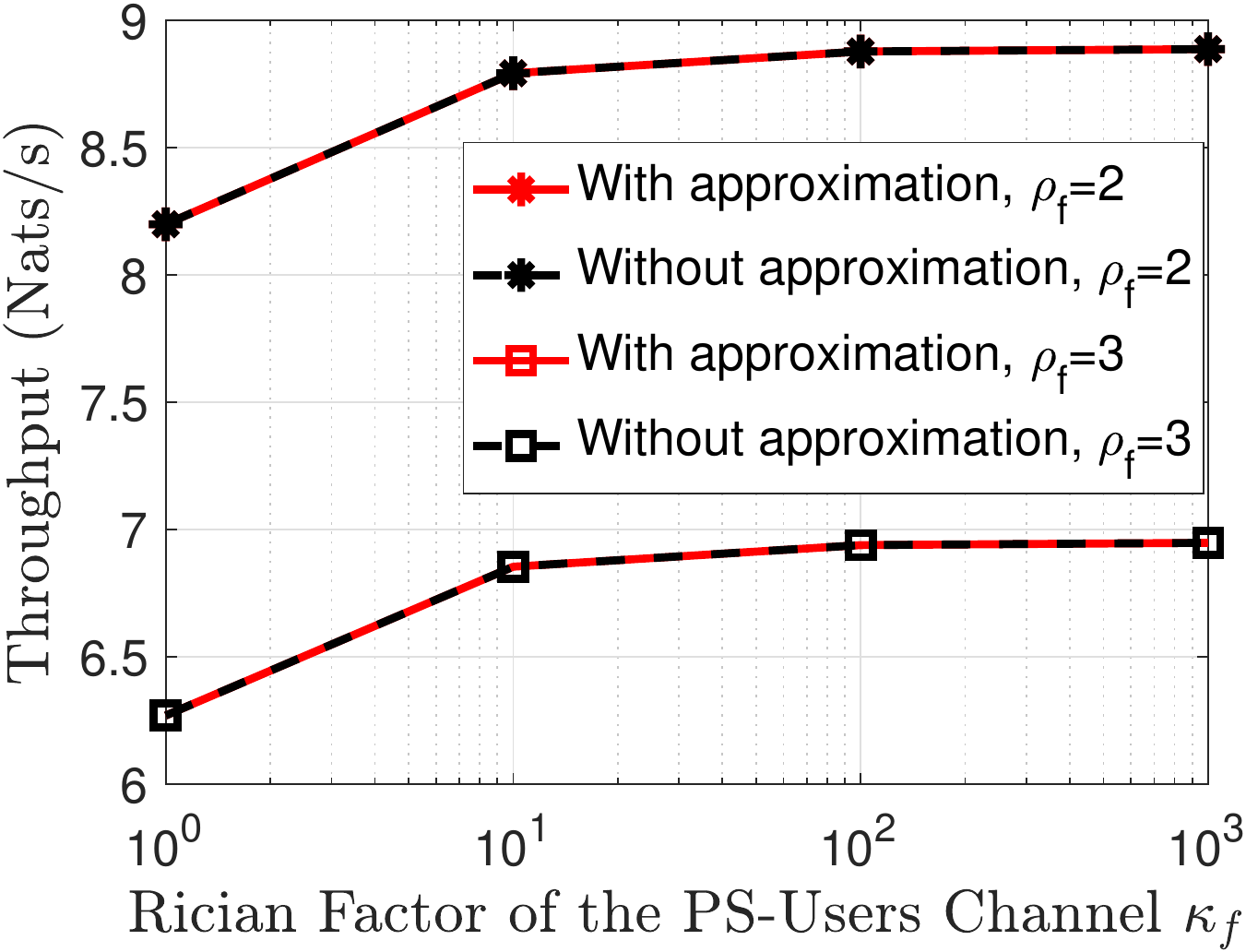}}
%\end{tabular}
%\end{tabularx}
\caption{ Validation of the rank-one approximation (a) throughput vs. number of PS antennas (b) throughput vs. Rician factor of the PS-users channel }
 \label{app_noapp}
\end{figure}
\begin{comment}
\begin{figure}[t!]
\centering
\includegraphics[width=2.3in]{appvsnoapp_Mp.eps}\hskip 11ex %fig1 does not have a caption
\subfloat{
\includegraphics[width=2.3in]{appvsnoapp_Ric.eps}    %fig2 has one
}
\caption{Throughput with and without rank-one approximation} \label{app_noapp}
\end{figure}
\end{comment}
Unless otherwise stated, the following set of parameters are used in all simulations: The number of users $K$ is assumed to be 10. The number of antennas at the PS and the AP is set as 5, i.e., $M_P=M_A=5$, and the number of elements at the IRS is set to be 25, i.e., $N=25$. The following coordinates are used for the IRS and the PS: $x_{\text{IRS}}=y_{\text{IRS}}=5$ m, $x_{\text{PS}}=30$ m. Also, the distance from the users to the PS is set as $r=10$ m. The channel-related parameters used in simulations are $\kappa_h=\kappa_f=\infty$, $\kappa_g=3$, $\kappa_c=0$, $\rho_h=\rho_f=2$, $\rho_g=2.8$, $\rho_c=3.5$. EH model parameters are obtained as $\eta=0.47$, $\xi=2.24 \times 10^{-5}$, and $p_{\text{sat}}=45$ mW. Maximum average power at the PS is set as $p_{\text{avg}}=1$ W and for the peak power at the PS we use $p_{\text{peak}}=2p_{\text{avg}}$. Backscatter coefficient is set as $\beta_i=0.6,~\forall i \in \mathcal{K}$ and the circuit power consumption for active IT is assumed to be $p_{c,i}=1$ mW, $\forall i \in \mathcal{K}$. The noise power spectral density is $-160$ dBm/Hz and the bandwidth is $1$ MHz. The stopping threshold for convergence is set as $\epsilon=0.001$ in all algorithms. The AP is assumed to apply MMSE receive beamforming for detecting users' active information signals. The results are based on the average of 1000 different channel realizations. 

\subsection{Numerical Results}
Fig. \ref{comp} compares the performance of the SA-based and AO-based methods for optimizing the reflection coefficients of the IRS elements when assisting in the users' backscatter transmission to the AP. The figure shows the network throughput versus the path-loss exponent of the channel between  AP and IRS for different number of IRS elements. According to Fig. \ref{comp}, the channel conditions between AP and IRS considerably affect the throughput. In specific, higher path-loss exponent results in greater attenuation for the signals that travel from the IRS to the AP, leading to lower SNR and throughput. Higher number of IRS elements can compensate for the poor channel conditions between AP and IRS. For instance, to achieve a throughput of $7$ Nats/s, employing $10$ elements at the IRS is sufficient when the path-loss exponent is $2.2$, while $20$ and $30$ elements are needed when the path-loss exponent is $2.9$ and $3.3$, respectively.  It can be seen that the SA-based method performs slightly better than the AO-based method; that's because in the SA-based method, the reflection coefficients of the IRS are jointly optimized in each iteration, while AO-based method optimizes each reflection coefficient individually. 
%review-solid red lines... modify according to the new figures.
%review-Fig. 6

Next, we show the tightness of the rank-one approximation in Algorithm \ref{alg_resource}. To this end, we compare the throughput of the proposed scheme to the calculated throughput when no rank-one approximation is performed. Fig. \ref{app_noapp} compares the maximized throughput with and without rank-one approximation. The solid red lines represent the throughput performance of our proposed scheme when Eigen-decomposition is performed in Algorithm \ref{alg_resource} for obtaining the beamforming vectors at the PS, while for the dotted black lines, the throughput is calculated based on $\Tilde{\boldsymbol{W}}_i^{*},~\forall i \in \mathcal{K}$ without extracting the beamforming vectors $\boldsymbol{w}_i^{*},~\forall i \in \mathcal{K}$. A very important takeaway from Fig. \ref{app_noapp} is the close match between the calculated throughputs with and without rank-one approximation which validates the accuracy of the Eigen-decomposition-based technique for extracting feasible rank-one matrices $\hat{\boldsymbol{W}}_i$ from $\Tilde{\boldsymbol{W}}_i^{*},~\forall i \in \mathcal{K}$.

Figs. \ref{pavg}-\ref{xIRS} assess the performance of our proposed IRS-empowered BS-WPCN by comparing it to four benchmark schemes labeled as "IRS-assisted WPCN", "Random IRS reflection matrices", "Pre-configured IRS reflection matrices", and "Equal time allocation". In particular, "IRS-assisted WPCN" is the scheme proposed in \cite{irs10}, where the conventional WPCN is upgraded by adding an IRS for assisting in downlink WET and uplink IT. To make comparisons fair, the co-located AP and PS (i.e., HAP) in \cite{irs10} is replaced by separate AP and PS as in our proposed scheme. Also, unlike \cite{irs10} which assumes single-antenna HAP, we modify the IRS-assisted WPCN model to have multi-antenna AP and PS. Finally, the IRS is assumed to have stable energy source, so, the time and power resources expended for energy collection of IRS in \cite{irs10} can now be used by IRS for assisting the network operations. "Random IRS reflection matrices" is the scheme in which the IRS phase shifts are randomly chosen in all $K+1$ time slots, while other variables are optimized based on the techniques discussed throughout the paper. In "Pre-configured IRS reflection matrices", the phase shifts are once selected randomly and the same selected phase shifts are used in all time slots and throughout the whole simulations. This resembles the scenario where IRS elements are pre-configured and cannot be reconfigured after IRS is deployed. Finally, "Equal time allocation" refers to the case with equal time slot length for all $K+1$ slots, i.e., the duration of each time slot is $\frac{1}{K+1}$ in this scheme. We also plot the graph for the proposed scheme after applying a 2-bit-resolution phase quantization to the optimized phase shifts of IRS elements in each time slot. Particularly, each optimized continuous phase shift is quantized to its closest value from the set $\mathcal{P}=\{\frac{\pi}{4},\frac{3\pi}{4},\frac{5\pi}{4},\frac{7\pi}{4}\}$, i.e., $\theta_{n,i}^{(q)}=\text{arg} \min_{\theta \in \mathcal{P}} |\theta - \theta_{n,i}^*|$, where $\theta_{n,i}^{(q)}$ is the quantized phase shift $\forall n \in N, i \in \mathcal{K} \cup \{K+1\}$. The throughput is then calculated based on the quantized phase shifts with all other variables taking the same optimized values as before.

As expected, the throughput improves with increasing the maximum average PS transmit power, number of IRS reflecting elements, and number of antennas at the AP and the PS (Figs. \ref{pavg}-\ref{Mp}). It is worth mentioning that although the EH circuits of some or all of the users enter the saturation region with increasing $p_{\text{avg}}$ in Fig. \ref{pavg}, the performance of the proposed model continues to improve with increasing the PS average transmit power. This is because the performance of backscatter communication is not constrained by the limitations of the EH circuits and the throughput obtained from backscatter transmission always gets better with more transmit power at the PS. On the other hand, the performance of "IRS-assisted WPCN" becomes saturated at $P_{\text{avg}}=30 ~\text{dBm}$ because the users cannot harvest more energy and the throughput, which is merely based on active wireless powered transmission, cannot be further enhanced. With increasing the x-coordinate of IRS in Fig. \ref{xIRS}, the throughput first decreases because of the longer distance between IRS and AP. However, the throughout begins to increase after some point since IRS gets closer to the users and the impinging signals on the IRS become stronger. This compensates for the increased distance between IRS and AP and improves the throughput. Hence, the IRS is better to be located either close to the AP or close to the users for achieving high network throughputs. 
\begin{comment}
\begin{figure*}[!htb]
   \begin{minipage}{0.48\textwidth}
     \centering
     \includegraphics[width=1.8in]{Pavg.eps}
     \caption{Throughput vs. average transmit}\label{pavg}
   \end{minipage}\hfill
   \begin{minipage}{0.48\textwidth}
     \centering
     \includegraphics[width=1.8in]{N.eps}
     \caption{Throughput vs. number of IRS elements}\label{N}
   \end{minipage}\hfill
   \begin{minipage}{0.48\textwidth}
     \centering
     \includegraphics[width=1.8in]{Ma.eps}
     \caption{Throughput vs. number of antennas}\label{Ma}
   \end{minipage}
\end{figure*}
\end{comment}
\begin{figure}[t!]
\centering
\includegraphics[width=1.8in]{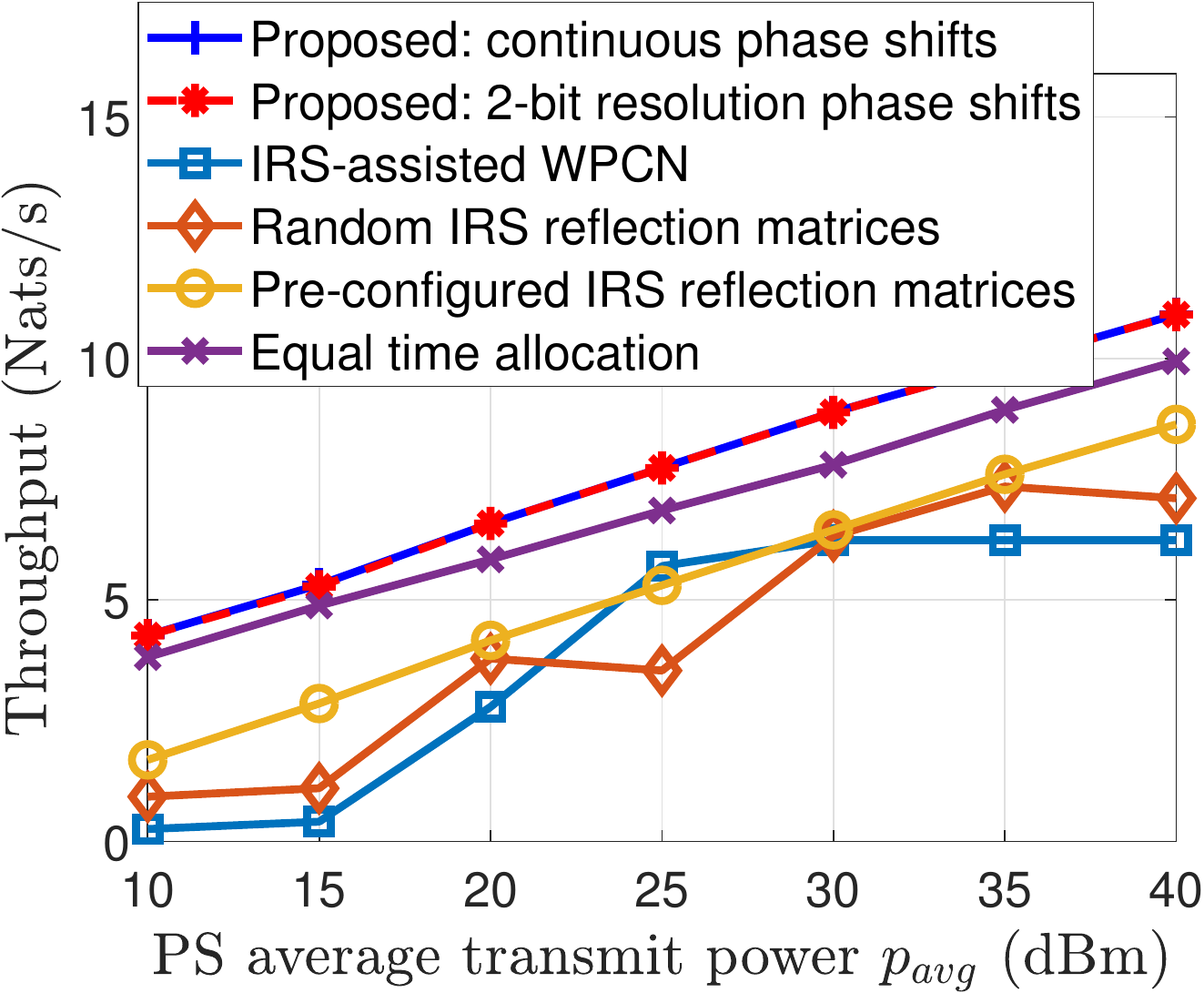}
\caption{ Throughput vs. average transmit power of the PS}
	\label{pavg}
\end{figure}

\begin{figure}[t!]
\centering
\includegraphics[width=1.7in]{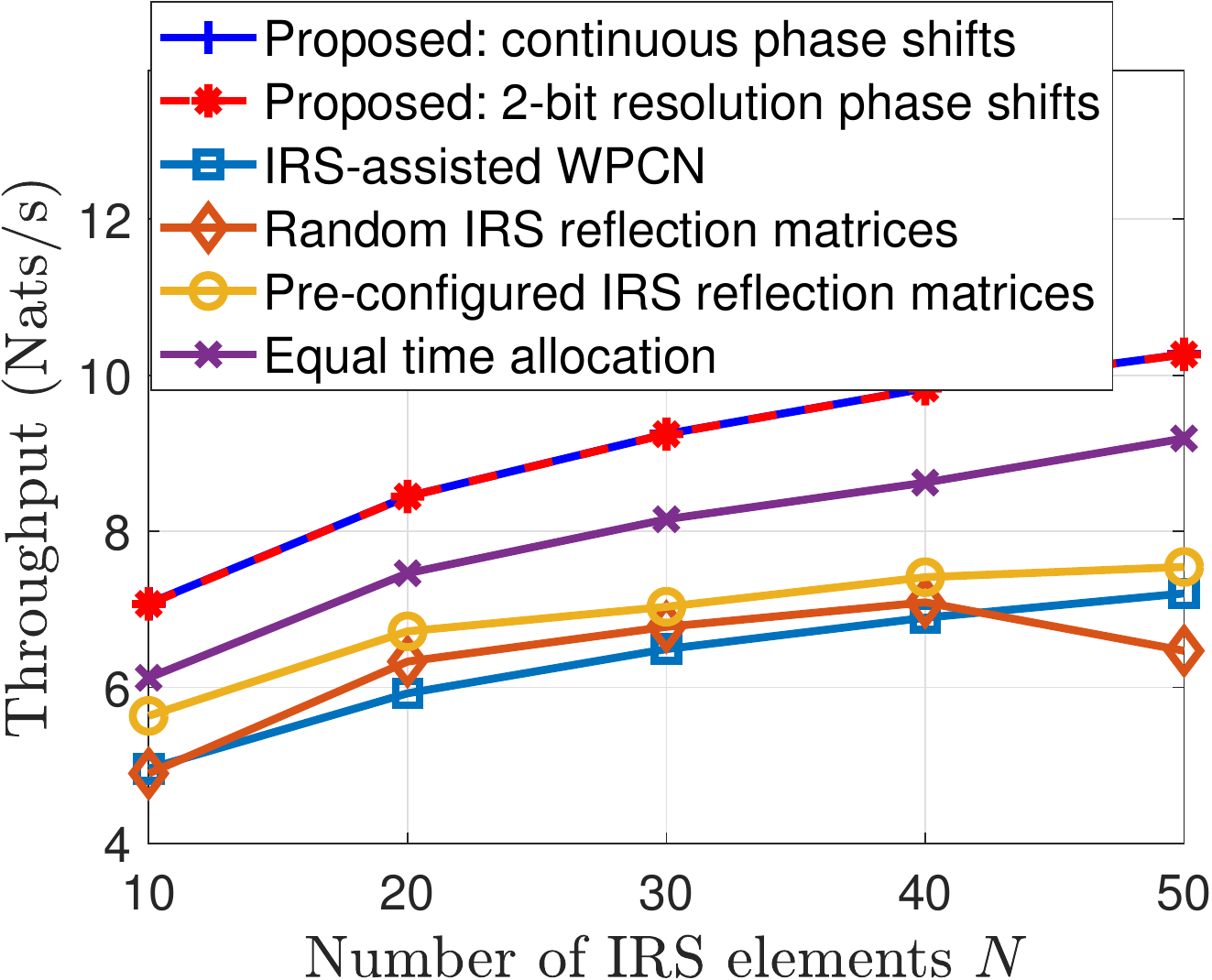}
\caption{ Throughput vs. number of IRS elements}
	\label{N}
\end{figure}

\begin{figure}[t!]
\centering
\includegraphics[width=1.7in]{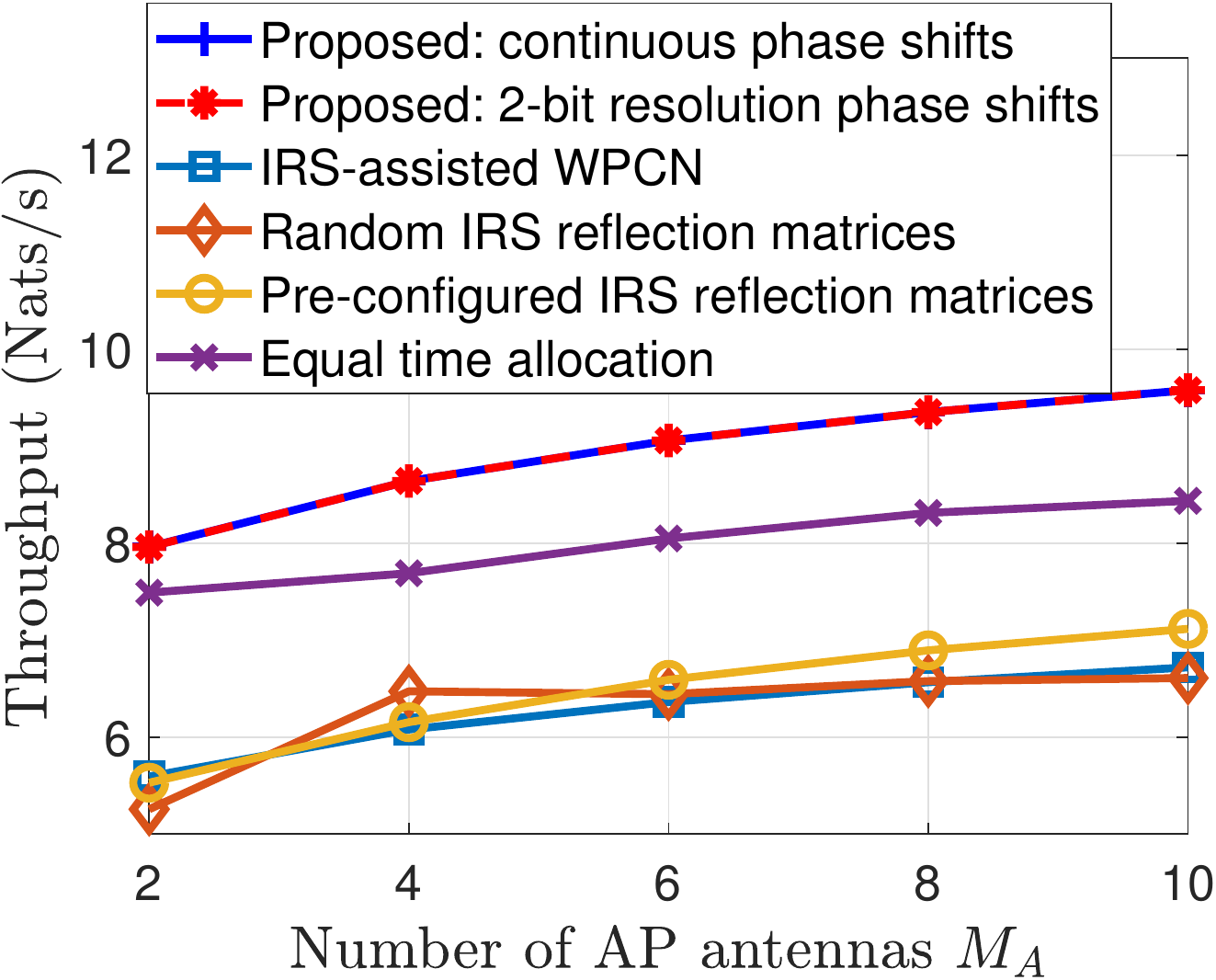}
\caption{ Throughput vs. number of antennas at the AP}
	\label{Ma}
\end{figure}
\begin{figure}[t!]
\centering
\includegraphics[width=1.7in]{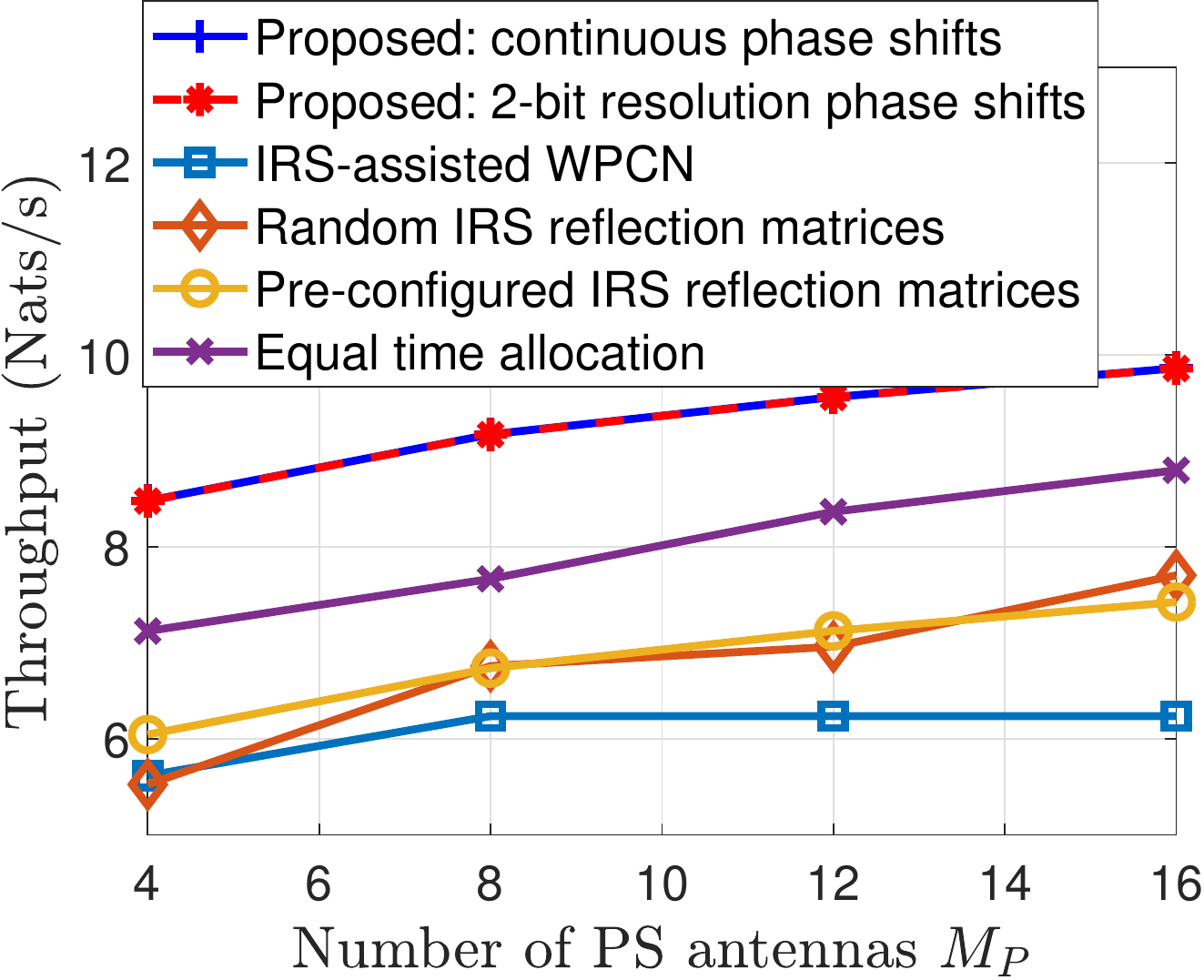}
\caption{ Throughput vs. number of antennas at the PS}
	\label{Mp}
\end{figure}

\begin{figure}[t!]
\centering
\includegraphics[width=1.7in]{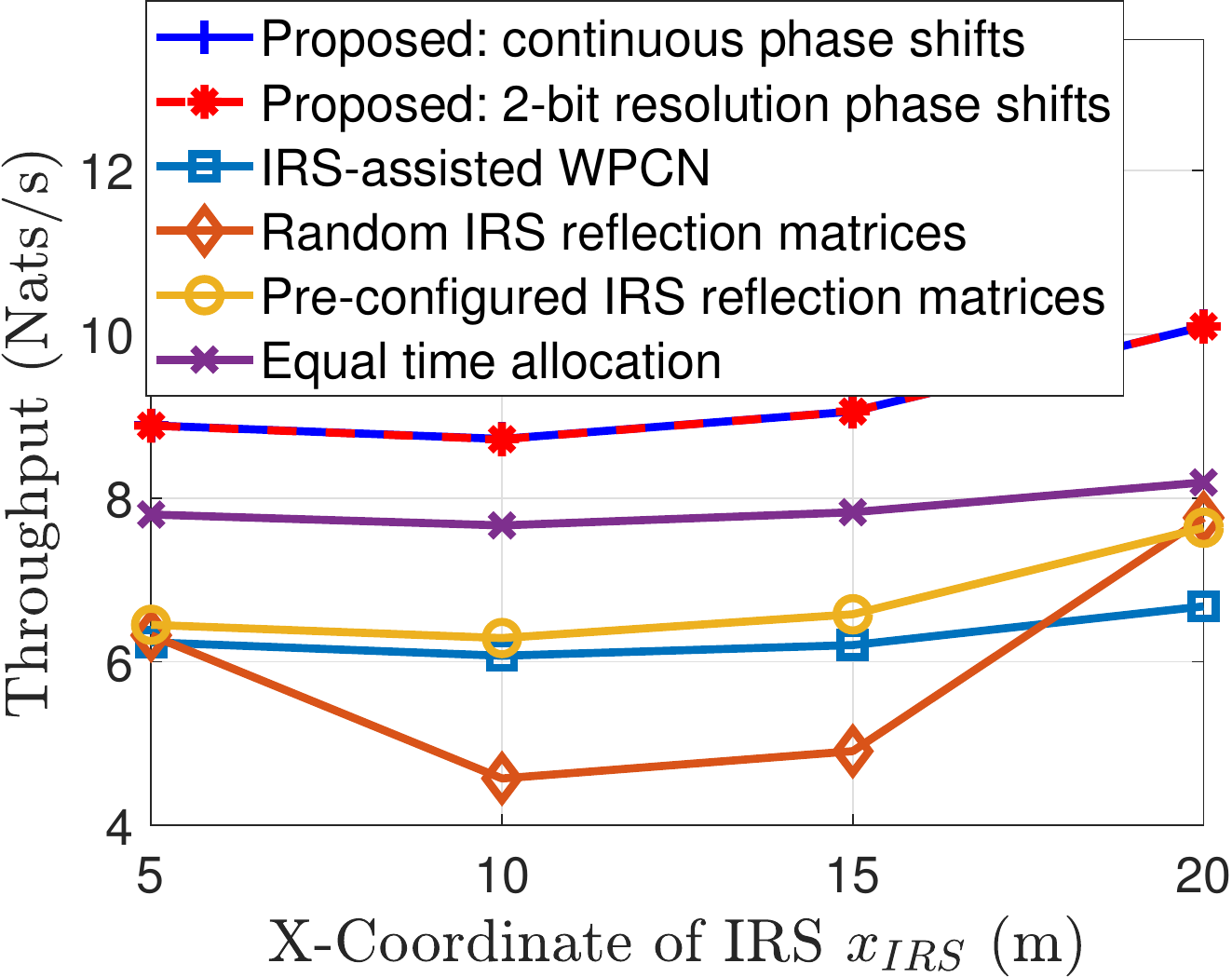}
\caption{Throughput vs. IRS x-coordinate}
	\label{xIRS}
	
\end{figure}

It is well observed that our proposed method performs remarkably better than the benchmark schemes, which endorses the efficiency of the proposed model and algorithms.  Specifically, our proposed IRS-empowered BS-WPCN considerably outperforms the IRS-assisted WPCN \cite{irs10} which demonstrates the potential of backscatter communication for improving the performance of WPCNs. Moreover, the performance of the proposed scheme is much superior to that of "Random IRS reflection matrices" and "Pre-configured IRS reflection matrices" schemes, showing the importance of real-time optimization and dynamic reconfiguration of IRS reflecting parameters. It can also be seen that the scheme with random selection of IRS phase shifts undergoes fluctuations, which is due to the fact that the randomly-chosen phase shifts may not result in constructive combination of the reflected signals with the signals of the direct path at the AP. In fact, with random phase shifts for IRS elements in each time slot, different levels of constructiveness/destructiveness can be expected for combination of signals and thus, a monotonic behavior for the "Random IRS reflection matrices" cannot be obtained. Finally, the gap between the proposed and "Equal time allocation" schemes confirms that optimizing time allocation is also important for enhancing the performance.

Another important observation is that the scheme with 2-bit-resolution quantized phase shifts almost overlaps with the one with continuous phase shifts, which indicates that our proposed algorithms can be applied to practical IRS-assisted systems with discrete phase shift values for the IRS elements. 

\section{Concluding Remarks}
This paper studied a multi-user BS-WPCN, where the backscatter and active IT of the users to the AP are assisted by an IRS. We investigated the optimization of IRS reflection coefficients, power allocation for the users' active IT, transmit beamforming of the PS, receive beamforming of the AP, and time allocation, for maximizing the total throughput of the network.  We presented a two-stage solution, where in the first stage, two methods based on AO and SA techniques have been proposed for optimizing the IRS reflection coefficients in the backscatter IT phase. In the second stage, we have used the SDR and SA techniques for optimizing AP transmit beamforming vectors, power and time allocation. Also, assuming an MMSE receiver at the AP, an efficient algorithm based on the BCD technique has been presented for jointly optimizing the AP receive beamforming vectors and IRS reflection coefficients when assisting the users' active IT. The accuracy and effectiveness of the proposed algorithms have been validated via numerical simulations. 

\section*{Appendix: Proof of Theorem 1}
Fixing $\boldsymbol{\omega}$ and $\boldsymbol{\Theta}_{K+1}$, \eqref{equiv} is convex with respect to $\boldsymbol{a}_i,~\forall i \in \mathcal{K}$, the optimal value of which can be obtained from the first-order optimality condition as

\begin{small}\begin{align}
\label{a_mmse}
\boldsymbol{a}_{i}^*=\sqrt{p_i}\Big(\sum_{j=1}^{K}p_j \boldsymbol{h}_j(\boldsymbol{\Theta}_{K+1})\boldsymbol{h}_j^H(\boldsymbol{\Theta}_{K+1})+\sigma^2\boldsymbol{I}_{M_A}&\Big)^{-1}\boldsymbol{h}_i(\boldsymbol{\Theta}_{K+1}),
\end{align}\end{small}
where the receive beamforming vector in \eqref{a_mmse} is the well-known MMSE receiver which minimizes the MSE as 
\begin{small}\begin{align}
\label{min_mse}
    E_{i,\text{min}}=1-p_i \boldsymbol{h}_i^H(\boldsymbol{\Theta}_{K+1}) J^{-1} \boldsymbol{h}_i(\boldsymbol{\Theta}_{K+1}),
\end{align}\end{small}with \begin{small}$J=\sum_{j=1}^K p_j \boldsymbol{h}_j (\boldsymbol{\Theta}_{K+1}) \boldsymbol{h}_j^H(\boldsymbol{\Theta}_{K+1})+\sigma^2 \boldsymbol{I}_{M_A}$.\end{small} Having $\{\boldsymbol{a}_i^*\}_{i=1}^K$ and $\boldsymbol{\Theta}_{K+1}$ fixed,  $\omega_i$ for minimizing  \eqref{equiv} is obtained as 
\begin{small}\begin{align}
\label{omeg}
    \omega_i^*=E_{i,\text{min}}^{-1}.
\end{align}\end{small}

Now, the problem for optimizing $\boldsymbol{\Theta}_{K+1}$ is obtained by substituting \eqref{a_mmse} and \eqref{omeg} into \eqref{equiv} as

\begin{small}\begin{align}
    \label{equiv2}
    \min_{\substack{ \boldsymbol{\Theta}_{K+1}}}~-\sum_{i=1}^K \log(E_{i,\text{min}}^{-1})&=\max_{\substack{ \boldsymbol{\Theta}_{K+1}}}\sum_{i=1}^K \log(E_{i,\text{min}}^{-1}) \\
\text{s.t.}&~\eqref{ThetaK+1}, \notag
\end{align}\end{small}

We drop the argument $\boldsymbol{\Theta}_{K+1}$ in the sequel for simplicity of notation. We have     
\begin{small}\begin{align}
\label{mse-sinr}
    &E_{i,\text{min}}^{-1}=\big(1-p_i\boldsymbol{h}_i^H J^{-1}\boldsymbol{h}_i \big)^{-1}=\Big(\dfrac{p_i\boldsymbol{h}_i^H J^{-1}\boldsymbol{h}_i -\big(p_i\boldsymbol{h}_i^H J^{-1}\boldsymbol{h}_i \big)^2}{p_i\boldsymbol{h}_i^H J^{-1}\boldsymbol{h}_i }\Big)^{-1}\notag\\&=\dfrac{p_i\boldsymbol{h}_i^H J^{-1}\boldsymbol{h}_i }{p_i\boldsymbol{h}_i^H J^{-1}\boldsymbol{h}_i -\big(p_i\boldsymbol{h}_i^H J^{-1}\boldsymbol{h}_i \big)^2}=1+\frac{\big(p_i\boldsymbol{h}_i^H J^{-1}\boldsymbol{h}_i \big)^2}{p_i\boldsymbol{h}_i^H J^{-1}\boldsymbol{h}_i -\big(p_i\boldsymbol{h}_i^H J^{-1}\boldsymbol{h}_i \big)^2}\notag\\&\overset{(\varpi_4)}{=}1+\frac{\big(p_i\boldsymbol{h}_i^H J^{-1}\boldsymbol{h}_i \big)^2}{p_i\boldsymbol{h}_i^H J^{-1}J J^{-1}\boldsymbol{h}_i -\big(p_i\boldsymbol{h}_i^H J^{-1}\boldsymbol{h}_i \big)^2}\notag\\&=1+\frac{\big(p_i\boldsymbol{h}_i^H J^{-1}\boldsymbol{h}_i \big)^2}{p_i\boldsymbol{h}_i^H J^{-1}\big(J-p_i\boldsymbol{h}_i\boldsymbol{h}_i^H\big)J^{-1} \boldsymbol{h}_i}\notag\\&=1+\frac{\big(p_i\boldsymbol{h}_i^H J^{-1}\boldsymbol{h}_i \big)^2}{p_i\boldsymbol{h}_i^H J^{-1}\big(\sum_{j\neq i}p_j\boldsymbol{h}_j\boldsymbol{h}_j^H+\sigma^2 \boldsymbol{I}_{M_A} \big)J^{-1} \boldsymbol{h}_i}\notag\\&\overset{(\varpi_5)}{=}1+\dfrac{p_{i}|\boldsymbol{a}_{i}^{*H} \boldsymbol{h}_i|^2}{\sum_{j \neq i} p_{j}|\boldsymbol{a}_{i}^{*H} \boldsymbol{h}_j|^2 + ||\boldsymbol{a}_{i}^{*H}||^2\sigma^2}=1+\gamma_{2,i,\text{mmse}},
\end{align}\end{small}where $\gamma_{2,i,\text{mmse}}$ is the SINR of $U_i$ for active IT, when MMSE receive beamforming is employed at the AP. In \eqref{mse-sinr}, $(\varpi_4)$ holds because $p_i\boldsymbol{h}_i^H J^{-1}J J^{-1}\boldsymbol{h}_i=p_i\boldsymbol{h}_i^H J^{-1}\boldsymbol{h}_i$ and $(\varpi_5)$ holds because $\boldsymbol{a}_i^*=\sqrt{p_i}J^{-1}\boldsymbol{h}_i$. The proof is completed by substituting \eqref{mse-sinr} into \eqref{equiv2}.

\end{document}